%% file: IJGS2019.tex
\documentclass[]{interact}

\usepackage{enumitem}
\usepackage{mdframed}
\usepackage{tabularx}
\usepackage{hyperref}
\usepackage{amssymb}
\usepackage{stmaryrd}
\usepackage{amsmath}
\usepackage{subfigure}
\usepackage{natbib}
\bibpunct[, ]{(}{)}{;}{a}{}{,}
\renewcommand\bibfont{\fontsize{10}{12}\selectfont}

\usepackage[ruled, vlined, linesnumbered]{algorithm2e} 

\SetCommentSty{mycommfont}

\usepackage{comment}
\usepackage{mathtools}
\usepackage{mathabx}
\usepackage{tkz-euclide}
\usepackage{comment}
\usepackage{wrapfig}
\usepackage{mathrsfs}

\theoremstyle{plain}
\newtheorem{theorem}{Theorem}[section]
\newtheorem{lemma}[theorem]{Lemma}

\newtheorem{proposition}[theorem]{Proposition}

\theoremstyle{definition}
\newtheorem{definition}[theorem]{Definition}
\newtheorem{example}[theorem]{Example}

\theoremstyle{remark}

\newtheorem{notation}{Notation}
\newtheorem{note}{Note}

\newcommand{\multistructure}[0]{{multistructure}}
\newcommand{\Multistructure}[0]{{Multistructure}}
\newcommand{\multisemilattice}[0]{{multisemilattice}}
\newcommand{\meetmultisemilattice}[0]{{meet-\multisemilattice}}
\newcommand{\joinmultisemilattice}[0]{{join-\multisemilattice}}
\newcommand{\multilattice}[0]{{multilattice}}

\newcommand{\minf}[0]{{\text{minf}}} 
\newcommand{\msup}[0]{{\text{msup}}}

\newcommand{\figref}[1]{Fig.~\ref{#1}}

\newcommand{\lfunction}[5]{#1 : #2 \to #3 , #4 \mapsto #5}

\newcommand{\empha}[1]{\textbf{#1}}

\newcommand{\minimalhandle}[0]{minimal-handle}
\newcommand{\maximalhandle}[0]{maximal-handle}
\newcommand{\minimumhandle}[0]{minimum-handle}
\newcommand{\maximumhandle}[0]{maximum-handle}

\setlist{noitemsep,topsep=1pt}

\author{
\name{Aimene Belfodil\textsuperscript{a, c}\thanks{Contact Aimene Belfodil. Email: aimene.belfodil@insa-lyon.fr}, Sergei O. Kuznetsov\textsuperscript{b}, Mehdi Kaytoue\textsuperscript{a, d}}
\affil{\textsuperscript{a}Univ Lyon, INSA Lyon, CNRS, LIRIS UMR 5205, F-69621, LYON, France; \textsuperscript{b}National Research University Higher School of Economics, Moscow, Russia; \textsuperscript{c}Mobile Devices Ing{\'{e}}nierie, 100 Avenue Stalingrad, 94800, Villejuif, France; \textsuperscript{d}Infologic, 99 avenue de Lyon, 26500 Bourg-L\`es-Valence, France;}
}

\title{On Pattern Setups and Pattern \Multistructure s}

\begin{document}

\maketitle

\input{sections/abstract}
\input{sections/introduction}
\input{sections/preliminaries}

\input{sections/patternSetups}
\input{sections/patternStructures}

\input{sections/fromClosedToSupportClosedDescriptions.tex}
\input{sections/multilattices.tex}
\input{sections/patternMultiStructures}
\input{sections/completion}
\input{sections/conclusion}
\bibliographystyle{tfcad}
\bibliography{biblio}   

\end{document}

%% file: sections/abstract.tex
\begin{abstract}
%
Modern order and lattice theory provides convenient mathematical tools for pattern mining, in particular for condensed irredundant representations of pattern spaces and their efficient generation. \emph{Formal Concept Analysis (FCA)} offers a generic framework, called \emph{pattern structures}, to formalize many types of patterns, such as itemsets, intervals, graph and sequence sets. Moreover, FCA provides generic algorithms to generate irredundantly all closed patterns, the only condition being that the pattern space is a meet-semilattice. This does not always hold, e.g., for sequential and graph patterns. Here, we discuss \emph{pattern setups} consisting of descriptions making just a partial order. Such a framework can be too broad, causing several problems, so we propose a new model, dubbed \emph{pattern multistructure}, lying between pattern setups and pattern structures, which relies on \emph{multilattices}. Finally, we consider some techniques, namely \emph{completions}, transforming pattern setups to pattern structures using sets/antichains of patterns.
\end{abstract}

%% file: sections/introduction.tex
\section{Introduction}

Modern order and lattice theory provide convenient mathematical tools for pattern mining, in particular for condensed irredundant representations of pattern spaces and their efficient generation.
Different formal tools has been proposed in the literature to model pattern spaces. Formal Concept Analysis (FCA - \cite{GanterW99}) has been proposed by \cite{wille1982restructuring} as a well-founded mathematical tool to models hierarchies of concepts related to some formal context (i.e. binary datasets).  While basic FCA provides a natural way to analyze binary datasets, datasets with more complex attributes (e.g. numerical or nominal ones) need to be transformed to such before any manipulation. This kind of transformation has been proposed by \cite{ganter1989conceptual} under the term of \emph{conceptual scaling} (i.e. \emph{binarizing}). Yet, even if \emph{conceptual scaling} is a quite general tool, \emph{binarizing} a dataset with regard to some pattern search spaces is not always obvious (i.e. \cite{DBLP:conf/cla/BaixeriesKN12}, \cite{DBLP:conf/ijcai/BelfodilKRK17}). In response to that, some other more natural tools to formalize complex pattern spaces has been proposed. One could cite Logical Concept Analysis (LCA) proposed by \cite{DBLP:conf/iccs/FerreR00} and Pattern Structures proposed by \cite{DBLP:conf/iccs/GanterK01}. 
Pattern Structures allow for instance to model in a quite natural way many pattern search spaces.  Indeed, itemsets, intervals \citep{DBLP:conf/ijcai/KaytoueKN11}, convex polygon \citep{DBLP:conf/ijcai/BelfodilKRK17}, partition \citep{DBLP:journals/amai/BaixeriesKN14} pattern spaces among others \citep{DBLP:conf/iccs/GanterK01, DBLP:conf/rsfdgrc/Kuznetsov09} can be modeled within the pattern structure framework.  
However, since pattern structures rely on meet-semilattices (i.e. conjunction of two patterns belongs to the pattern search spaces), some pattern spaces that are only partially ordered sets (posets) cannot be ``directly'' defined using such a framework. 

Consider for instance the example dataset depicted in \figref{fig:introexample} containing $4$ objects described by attribute \texttt{"value"} and labeled positive or negative. 
We are interested by the task of finding ``good" rules $d \to +$ in this dataset with $d$ a description given by attribute \texttt{value}. Rather than considering the usual meet-semilattice of intervals as the one proposed by \cite{DBLP:conf/ijcai/KaytoueKN11}; descriptions $d$ are restrained to open intervals of the form $(v]$ and $[v)$ or singleton $\{v\} \subseteq \mathbb{R}$ (see \figref{fig:introexample} - right).
Patterns (descriptions) form together a poset $(\mathcal{D}, \supseteq)$ where $\supseteq$ is the \emph{subsumption order} (i.e. if $d_1$ subsumes $d_2$ then if pattern $d_2$ holds for an object $g$ then pattern $d_1$ hold to). However, $(\mathcal{D}, \supseteq)$ does not form a meet-semilattice. For instance, the set $\{\{3\},\{5\}\}$ does not have a meet, since lower bounds of $\{\{3\},\{5\}\}$ have two maximal elements w.r.t. $\supseteq$ (\emph{i.e.} $[3)$ and $(5]$). Hence, the description space does not induce a pattern structure~\citep{DBLP:conf/iccs/GanterK01}. It does form actually a \emph{pattern setup} (\cite{DBLP:conf/cla/LumpeS15}) which rely on a simple poset with no additional properties.

Description spaces like the one depicted in \figref{fig:introexample} are numerous. For instance, sequence of itemset patterns (\cite{agrawal1995mining}) ordered by \emph{``is subsequence of"} do not form a meet-semilattice \citep{DBLP:journals/ml/Zaki01, DBLP:conf/icfca/CodocedoBKBN17}. Sequential meet-semilattice in FCA \citep{DBLP:journals/ijgs/BuzmakovEJKNR16, DBLP:conf/icfca/CodocedoBKBN17} refers usually to \emph{set (i.e. conjunction) of sequences} rather than to the poset of \emph{sequences} itself. Same holds for the graph meet-semilattice from~\cite{DBLP:conf/pkdd/Kuznetsov99}. In general, the base pattern setup is transformed to a pattern structure using sets (i.e. conjunctions) of descriptions thus providing a richer pattern language. Such transformations are naturally called \emph{completions}, which refers to well-known completions in ordered sets and lattice theory, like Dedekind-McNeille, Alexandroff or Antichain completions~\citep{davey2002introduction, BoldiVigna2016}. For instance, in \figref{fig:introexample}, set of patterns $\{\texttt{value} \geq 3, \texttt{value} \leq 5\}$ which is equivalent to $3 \leq \texttt{value} \leq 5$ belongs to the antichain completion of $(D, \sqsubseteq)$ but and not to the base description language $(D, \sqsubseteq)$. 



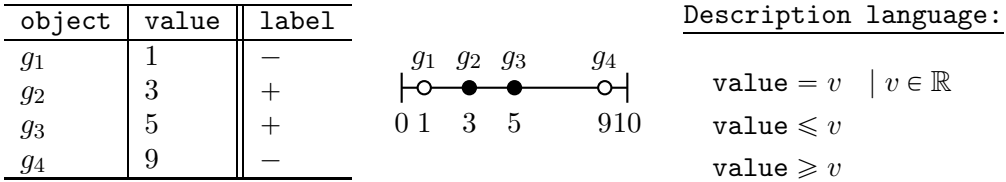
\begin{figure}[t]
{ 
  \begin{minipage}{0.36\linewidth} 
  \centering 
  \begin{tabular}{l|l||l}
   \hline   
   \texttt{object} & \texttt{value} & \texttt{label}\\
   \hline
   $g_1$ & $1$ & $-$\\ 
   $g_2$ & $3$ & $+$\\
   $g_3$ & $5$ & $+$\\
   $g_4$ & $9$ & $-$\\
   \hline 
  \end{tabular}
\end{minipage}%
 \begin{minipage}{0.28\linewidth} 
  \centering 
  \scalebox{1}{\input{figures/intervalpatternExample1D.tex}}
 \end{minipage}
 \begin{minipage}{0.30\linewidth}
 \underline{\texttt{Description language:}}
 \begin{eqnarray*}
 \texttt{value} = v & \mid v \in \mathbb{R}\\
 \texttt{value} \leq v &\\
 \texttt{value} \geq v &
 \end{eqnarray*}
 \end{minipage} 
}
\caption{Dataset with numerical attribute (left), its representation in $\mathbb{R}$ with black dots representing positive objects (center) and the description language (right)\label{fig:introexample}.}
\end{figure} 

Understanding properties of pattern setups independently from their completions is fundamental for answering many practical questions. For instance, consider the question \textit{``What are the best descriptions covering all positive instances?"}. If \emph{better} stands for \emph{more relevant than} as in relevance theory \citep{DBLP:journals/jmlr/GarrigaKL08}, the answer will be the two best incomparable rules $\texttt{value} \geq 3 \to +$ and $\texttt{value} \leq 5 \to +$ rather than only one in the completion $3 \leq \texttt{value} \leq 5 \to +$. More generally, we look for \emph{maximal common patterns} (i.e. \emph{support-closed patterns} as called by \cite{DBLP:journals/tcs/BoleyHPW10}) of the positive instances which could be multiple.

\newpage

\noindent
\textbf{Outlines and Contributions.} In this paper we present the following results:
\begin{enumerate}[topsep=1pt]
    \item In Section \ref{sec:patternSetups}, we start by studying pattern setups that was proposed by \cite{DBLP:conf/cla/LumpeS15} as a tool to models pattern search spaces relying only on posets providing hence a better-understanding of this wider-class of framework. Afterwards, we revisit briefly the notion of pattern structures in Section \ref{sec:patternStructure}.
    \item In Section \ref{sec:supportclosed}, we point-out the major problem related to pattern setups. In fact, since pattern setups rely only on posets without no additional properties, they are very permissive. Simply put, given some set of objects, they could share some common descriptions in the pattern space but none of them is maximal w.r.t. the subsumption order. We show that this problem is directly linked to the fact that the considered poset is not a multilattice. 
    \item In Section \ref{sec:multilattices}, we present multilattices and we show that all (doubly) chain-complete posets are complete multilattices but that the converse does not hold.
    \item In Section \ref{sec:pattermultistructure}, we propose the framework of pattern \multistructure s. In a nutshell, analogously to pattern structures, pattern \multistructure s are based on multilattices. Such a structure provides the fact that covering descriptions of a set of objects are deducible from the maximal common descriptions using the subsumption order. This does not necessarily hold in an arbitrary pattern setup.
    \item Next, in Section \ref{sec:completions} we revisit completion (i.e. transformation) of pattern setups to pattern structures and we show that the usual completion using antichain of patterns, namely \emph{antichain completion}, induce a pattern structure if and only if the pattern setup is a pattern \multistructure{}. 
    
    \item Finally, we wrap-up in section \ref{sec:conclusiondiscussion} and we discuss some open problems related particularly to the enumeration of extents (i.e. subset of objects that are separable in the description language) in an arbitrary pattern setup. 
    
\end{enumerate}

Please note that this paper is a thorough extension of a first paper published in CLA'18 \citep{DBLP:conf/cla/BelfodilKK18}. The main difference with CLA'18 paper is more mathematical results and much more detailed explanations of the framework of pattern setups, pattern \multistructure s and their completions. Note that in this first paper, pattern \multistructure s were called pattern hyper-structures and were linked with the notion of hyper-lattices briefly investigated by \cite{DBLP:journals/ml/Zaki01}. Stefan Schmidt attracted our attention to the work of \cite{benado1955ensembles} on multilattices that interestingly was directly linked to pattern \multistructure{} making this newly proposed framework more mathematically founded. 

%% file: figures/intervalpatternExample1D.tex
\begin{tikzpicture}
[bdot/.style = {circle, fill=black, draw=black,inner sep=0pt, minimum size=5pt},
wdot/.style = {circle, fill=white, draw=black,inner sep=0pt, minimum size=5pt},
thick]
                        	
	\draw ( 0,0.2) -- + (0,-0.4) node[below] {$0$};
	\draw (3,0.2) -- + (0,-0.4) node[below] {$10$};
	\draw[thick] (0,0) -- node[below=2mm] {} + (3,0);

	\tkzDefPoint(0.3*1,0){g1}
    \tkzDefPoint(0.3*3,0){g2}
    \tkzDefPoint(0.3*5,0){g3}
    \tkzDefPoint(0.3*9,0){g4}
	
	\foreach \i in {1, 4}
	{
    	\node[wdot,label={$g_\i$}] at (g\i) {};
	}
	\foreach \i in {2, 3}
	{
    	\node[bdot,label={$g_\i$}] at (g\i) {};
	}
	\tkzText[below=2mm](g1){1}
	\tkzText[below=2mm](g2){3}
	\tkzText[below=2mm](g3){5}
	\tkzText[below=2mm](g4){9};
\end{tikzpicture}

%% file: sections/preliminaries.tex
\section{On Partially-Ordered Sets\label{sec:preliminaries}}
To make this work self-contained, in this section we recall basic definitions and results from order theory and introduce our notations that are largely inspired by \cite{GanterW99,davey2002introduction,roman2008lattices}.  
In this paper, we will be using the following base notations:
\begin{itemize}[noitemsep,topsep=1pt]
    \item For any set $P$, the set $\wp(P)$ denotes the powerset of $P$. 
    \item For any mapping $f: E \to F$ and any subset $A \subseteq E$, the set $f[A]$ denotes the image of $A$ w.r.t. $f$, that is:
    \begin{eqnarray*}
    f[A] = \left\{f(a) \mid a \in A \right\}
    \end{eqnarray*}
    \item The set of all mappings $f : E \to F$ is denoted by $F^E$. 
    
\end{itemize}

\subsection{Basic Definitions}

\begin{definition}
A \empha{partial order} on a set $P$ is a binary relation $\leq$ on $P$ that is  
\empha{reflexive} ($(\forall x \in P)$ $x \leq x$), \empha{transitive} ($(\forall x,y,z \in P)$ if $x \leq y$ and $y \leq z$ then $x \leq z$.) and \empha{antisymmetric} ($(\forall x,y,z \in P)$ if $x \leq y$ and $y \leq x$ then $x=y$).
%
The pair $(P,\leq)$ is called a \empha{partially ordered set} or a \empha{poset} for short. Two elements $x$ and $y$ from $P$ are said to be \empha{comparable} if $x \leq y$ or $y \leq x$; otherwise, they are said to be \empha{incomparable}. A subset $S \subseteq P$ is said to be a \empha{chain} (resp. an \empha{antichain}) if all elements of $S$ are pairwise comparable (resp. incomparable). The set of all chains (resp. antichains) of $P$ is denoted by $\mathscr{C}(P)$ (resp. $\mathscr{A}(P)$).
\end{definition}

\begin{note}
(Finite) posets are generally presented using the so called \empha{Hasse Diagram}. Where elements of the poset are represented on the plane from the smallest ones to the largest ones. Only direct neighbors (i.e. two elements are said to be direct neighbors if there is no elements strictly lying between them in the poset) are linked by a segment of line. The poset can indeed be deduced from this diagram by adding reflexivity and transitivity.
\figref{fig:hassediagramSimple} \textbf{(1)} depicts the Hasse Diagram of the powerset of a set $E = \{a,b,c\}$ ordered by set inclusion (i.e. poset $(\wp(E), \subseteq)$). 

Through some misuse of the notion, we will present in this paper some infinite posets using the Hasse Diagram. \figref{fig:hassediagramSimple} \textbf{(2)} presents a poset $(P, \leq)$ where:
\begin{itemize}
\item $P = \{\bot,\top,b_0,b_1\} \cup \{a_i \mid i \in \mathbb{N}\}$,
\item  $\bot \leq b_0 \leq b_1 \leq \top$, and 
\item $(\forall i \in \mathbb{N})$ $\bot \leq a_i \leq a_{i+1} \leq \top$.
\end{itemize}

\end{note}

In what follows, $(P, \leq)$ denotes a poset and $S \subseteq P$ denotes an arbitrary subset. 

\begin{definition}
The \empha{principal ideal} (resp. \empha{principal filter}) of $p \in P$, denoted by $\downarrow p$ (resp. $\uparrow$) is the set of all elements below (resp. above) it:
\begin{eqnarray*}
\downarrow p = \{x \in P \mid x \leq p\} & & \uparrow p = \{x \in P \mid p \leq x\}
\end{eqnarray*}
\end{definition}

\begin{definition}
A set $S \subseteq P$ is said to be a \empha{lower ideal} or a \empha{downset} if: 
\begin{align*}
(\forall x \in P) \:\: \left( (\exists s \in S) \: \: x \leq s \Rightarrow x \in S\right)
\end{align*}
The notion of \empha{upper ideal} (\empha{upset}) is defined dually.   
The set of all lower (resp. upper) ideals of $(P, \leq)$ is denoted $\mathscr{O}(P)$ (resp. $\mathscr{U}(P)$). 
\end{definition}

\begin{definition}
The \empha{down closure} (resp. \empha{up closure}), denoted by $\downarrow S$ (resp. $\uparrow S$) associates to a subset $S$ the smallest lower-ideal (resp. upper-ideal) enclosing it:
\begin{eqnarray*}
\downarrow S = \{x \in P \mid (\exists s \in S)\:x \leq s\}  = \bigcup_{s \in S} \downarrow s &  & \uparrow S = \{x \in P \mid (\exists s \in S)\:s \leq x\} = \bigcup_{s \in S} \uparrow s
\end{eqnarray*}
\end{definition}

\begin{definition}
An element $p \in P$ is said to be a \empha{lower bound} (resp. \empha{upper bound}) of $S$ if it is below (resp. above) all elements of $S$. 
The set of lower (resp. upper) bounds of $S$ in $P$, denoted by $S^\ell$ (resp. $S^u$), is the \emph{lower ideal} (resp. \emph{upper ideal}) given by:
\begin{eqnarray*}
S^\ell = \{x \in P \mid (\forall s \in S)\:x \leq s\}  = \bigcap_{s \in S} \downarrow s  &  & S^u = \{x \in P \mid (\forall s \in S)\:s \leq x\} = \bigcap_{s \in S} \uparrow s
\end{eqnarray*}
\end{definition}

\begin{note}
Note that $\uparrow \emptyset = \downarrow \emptyset = \emptyset$ and $\emptyset^\ell = \emptyset^u = P$.  
\end{note}

\begin{figure}[t]
\centering
\hfill
\begin{minipage}{0.30\linewidth}
\scalebox{1}{\input{figures/hasse_diagram_powerset_abc.tex}}
\end{minipage}
\hfill
\begin{minipage}{0.20\linewidth}
\scalebox{1}{\input{figures/hasse_diagram_minimum_minimal.tex}}
\end{minipage}
\hfill
\begin{minipage}{0.20\linewidth}
\scalebox{1}{\input{figures/hasse_diagram_substringofabandba.tex}}
\end{minipage}
\hfill
\begin{minipage}{0.20\linewidth}
\scalebox{1}{\input{figures/hasse_diagram_impossible_extents.tex}}
\end{minipage}
\hfill
\caption{
 \textbf{From left to right:} \textbf{(1)} Poset  $(\wp(\{a,b,c\}), \subseteq)$. \textbf{(2)} A poset $(P, \leq)$ with $P = \{\bot,\top,b_0,b_1\} \cup \{a_i \mid i \in \mathbb{N}\}$, $\bot \leq b_0 \leq b_1 \leq \top$ and $(\forall i \in \mathbb{N})$ $\bot \leq a_i \leq a_{i+1} \leq \top$. \textbf{(3)} Poset of common substrings of $``ab"$ and $``ba"$.
 \textbf{(4)} A subposet of $(\wp(\{g_1,g_2,g_3,g_4\}), \subseteq)$.
}
\label{fig:hassediagramSimple}
\end{figure}
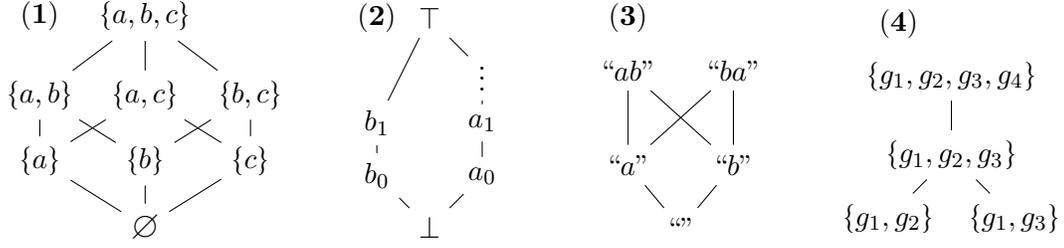

\subsection{Minimal, Minimum, Infimum, ...}

Informally speaking, not all elements in a subset $S$ have the same ``position". Some elements $x$ are in the ``interior" of $S$, that is $\exists y,z \in S$ such that $y < x < z$; while some others are on its ``border", that is there is no element in $S$ strictly below/above them. These border elements are formally defined in Definition \ref{def:maximalminimal}.

\begin{definition}
\label{def:maximalminimal}
An element $s \in S$ is said to be a \empha{minimal} (resp. \empha{maximal}) element in $S$ if all its strict lower (resp. upper) bounds are outside $S$.  
The set of \emph{minimal} (resp. \emph{maximal}) elements of $S$ is denoted by $min(S)$ (resp. $max(S)$) and given by:
\begin{eqnarray*}
min(S) = \{x \in S \mid\:\downarrow\:x \cap S = \{x\}\} & \phantom{:::} & 
max(S) = \{x \in S \mid\:\uparrow\:x \cap S = \{x\}\}
\end{eqnarray*}
\end{definition}

Note that $min(S)$ and $max(s)$ are \emph{antichains}. Moreover, $min(\emptyset) = max(\emptyset) = \emptyset$ by definition. 
An intuitive and important relationship between the minimal elements of a subset and the minimal elements in its up-closure is presented lemma \ref{lemma:propertiesminandarrows}.

\begin{lemma}\label{lemma:propertiesminandarrows}
We have
$min(\uparrow S) = min (S)$  and  $max(\downarrow S) = max(S)$ 
\end{lemma}

\begin{mdframed}
\input{proofs/proof1.tex}
\end{mdframed}

\begin{definition} 
\label{def:minimummaximum}
An element $m \in S$ is said to be \empha{minimum} or a \empha{smallest element} of $S$ if it is a lower bound of $S$. Formally:  
$
m \in S \textit{ and } (\forall s \in S) \:\: m \leq s
$.
The notion of \empha{maximum} or \empha{greatest element} of $S$ is defined dually.


\end{definition}

\begin{note}
One should notice that not all subsets have a minimum or a maximum (see Example~\ref{example:maximumdoesnotexist}). Moreover, the minimum and the maximum are unique if they exist.
\end{note}

\begin{example}
\label{example:maximumdoesnotexist}
Consider again \figref{fig:hassediagramSimple} \textbf{(1)} and the subset $S = \{\{a\}, \{a,b\}, \{a,c\}\}$. It is clear that $min(S) = \{a\}$ and $max(S) = \{\{a,b\}, \{a,c\}\}$. Obviously, subset $S$ has a minimum which is $\{a\}$. However, it does not have a maximum since no element in $S$ is above all the other elements.
\end{example}

\begin{note}
If a poset has a maximum (resp. minimum) element, we say that the poset is \empha{upper-bounded} (resp. \empha{lower-bounded}) and this element is called the \empha{top} (resp. \empha{bottom}) and is denoted $\top$ (resp. $\bot$). 
%
%
A poset is said to be \empha{bounded} if it is both lower-bounded and upper-bounded.
\end{note}

It is very important to distinguish between the minimal elements and the minimum.  In fact, 
``if a subset has a minimum then it has a unique minimal element". However, the converse of this statement is not true, that is even if a subset have a single minimal element, it can have no minimum. Indeed, the intuition that every element in a subset $S$ is at least above one minimal element in $min(S)$ does not always hold in infinite posets. Such a property is presented in Definition \ref{def:minimalmaximalhandle}.

\begin{definition}[\textbf{Definition 2.4 in \cite{DBLP:journals/dm/MartinezGGC05}}]
\label{def:minimalmaximalhandle}
We say that:
\begin{itemize}
\item $S$ is \empha{\minimumhandle{}} if $S$ has a minimum (i.e. $(\exists m \in S)$ $S \subseteq \uparrow m$).
\item $S$ is \empha{\maximumhandle{}} if $S$ has a maximum (i.e. $(\exists m \in S)$ $S \subseteq \downarrow m$).

\item $S$ is \empha{\minimalhandle{}} if $S \subseteq \uparrow min(S)$ 
\item $S$ is \empha{\maximalhandle{}} if $S \subseteq \downarrow max(S)$.
\end{itemize}
\end{definition}

\begin{example}
Consider the poset $(P, \leq)$ depicted in \figref{fig:hassediagramSimple} and let $S = P \backslash \{\top\} = \{\bot, b_0, b_1\} \cup \{a_i \mid i \in \mathbb{N}\}$. It is clear that $max(S) = \{b_1\}$. Since $\downarrow \{b_1\} = \{\bot, b_0, b_1\}$, subset $S$ is not a \maximalhandle{}. Hence, even if $S$ has a single maximal element, it has no maximum. On the other hand, $min(S) = \bot$. Since $S \subseteq \uparrow \bot = P$, we can say that $S$ is a \minimalhandle{}. Moreover, since the minimal element is unique than $S$ is \minimumhandle{} which minimum is $\bot$.
\end{example}

\begin{note}
If $S$ is minimal-handle, then $S$ has a minimum if and only if $S$ has a unique minimal element. It is clear that all subsets of a \empha{finite posets} (i.e. posets with a finite set) are maximal-handle and minimal-handle. Posets where all subsets are \minimumhandle{} are said to be \empha{well-founded} or equivalently have the \empha{minimal condition} or the \empha{descending chain condition (DCC)}. Dually, posets where all subsets are \maximalhandle{} are said to be \empha{dually well-founded} or equivalently have the \empha{maximal condition} or the \empha{ascending chain condition (ACC)}. A poset having at the same time the \textbf{ACC} and \textbf{DCC} is said to be \empha{chain-finite} since it has no infinite chain (but, still, could have infinite antichains). A poset is in fact finite if and only if it is \empha{chain-finite} and \empha{antichain-finite}.    
\end{note}

\begin{definition} 
\label{def:infimumsupremum}
The largest lower bound of $S$ (i.e. the maximum of $S^\ell$) if it exists is called the  \empha{infimum} or the \empha{meet} of $S$ and is denoted $inf(S)$ or $\bigwedge S$. The \empha{join} or the \empha{supremum} of $S$ is given by the minimum of $S^u$ and is denoted $sup(S)$ or $\bigvee S$.
%
\end{definition}

Having a minimum is a very strong property, in fact if a subset has a minimum then the minimum is also the infimum. However, the converse is not true. Indeed a subset can have an infimum without having a minimum. For instance, consider \figref{fig:hassediagramSimple} \textbf{(2)}, it is clear that the chain $S = \{a_i \mid i \in \mathbb{N}\}$ does not have a maximum. However, $S$ still have a supremum which is $\top$. Indeed, $S^u = \{\top\}$ which minimum is $\top$.

One interesting property is stated in the following Lemma.

\begin{lemma}
\label{lemma:lowerboundsarestables}
Let $(P, \leq)$ be a poset, $S \subseteq P$, we have:
\begin{itemize}
    \item For any $A \subseteq S^\ell$. If $A$ has a join $\bigvee A \in P$ then $\bigvee A \in S^\ell$.
    
    \item For any $A \subseteq S^u$. If $A$ has a meet $\bigwedge A \in P$ then $\bigwedge A \in S^u$.
\end{itemize}

\end{lemma}

\begin{mdframed}
\input{proofs/proof2.tex}
\end{mdframed}

\begin{note}
\label{note:reversemeetandjoin}
One should note that, in case of existence, we have:
\begin{eqnarray*}
\bigwedge S = \bigvee S^\ell & & \bigvee S = \bigwedge S^u 
\end{eqnarray*}

Moreover, since for a poset $(P, \leq)$ we have $\emptyset^\ell = \emptyset^u =  P$, then the empty set has a meet (resp. join) if and only if the poset is upper-bounded (resp. lower bounded) and we have
$\bigwedge \emptyset = \bigvee P = \top$ and $\bigvee \emptyset = \bigwedge P = \bot$. 
\end{note}

\subsection{Lattices}

\begin{definition}
\label{def:lattice}
A poset $(P, \leq)$ is said to be: 
\begin{itemize}
    \item A \empha{meet-semilattice} if for all nonempty finite subsets $S \subseteq P$, $S$ has its meet.  
    \item A \empha{join-semilattice} if for all nonempty finite subsets $S \subseteq P$, $S$ has its join.  
    \item A \textbf{lattice} if it is both meet-semilattice and join-semilattice.
    \item A \textbf{complete lattice} if all its subsets including $\emptyset$ has their meet and join.
\end{itemize}
\end{definition}

\begin{note}
For any set $E$, the poset $(\wp(E), \subseteq)$ is a complete lattice where the meet is set intersection $\bigcap$ and the join is set union $\bigcup$. Such a poset is called a \empha{powerset lattice}. 
\end{note}

\begin{example}
Consider posets depicted in Posets  \figref{fig:hassediagramSimple}. Poset \textbf{(1)} is a powerset lattice on $\{a,b,c\}$. Hence, it is a complete lattice. Poset \textbf{(2)} is also a complete lattice where the join of any infinite subset is $\top$. Poset \textbf{(3)} of common substrings of \textbf{``ab"} and \textbf{``ba"} ordered by \emph{is substring of} is neither a meet-semilattice nor a join-semilattice. Indeed, $\{``ab",``ba"\}$ has not an infimum since $\{``ab",``ba"\}^\ell = \{``a", ``b", ``"\}$ has two maximal elements and thus no maximum. 
Dually, subset $\{``a",``b"\}$ has not a join since $\{``a",``b"\}^u = \{``ab", ``ba"\}$ has two minimal elements and thus no minimum.
\end{example}

\begin{note}
\label{note:latticesproperty}
Meet-semilattices have a weaker, yet equivalent, definition characterizing them. In fact, to check if a poset is a lattice, one should only check if:
\begin{eqnarray*}
(\forall x,y \in P) \:\:\: \{x,y\} \text{ has a meet } & \textbf{ and } & (\forall x,y \in P) \:\:\: \{x,y\} \text{ has a join }
\end{eqnarray*}

That is if all pair of elements have their meets (resp. join) then all nonempty finite subsets of the poset have their meet (resp. join). 
Another important remark that is related to Note~\ref{note:reversemeetandjoin} is the fact that \underline{all complete semilattices are complete lattices}. 

One should notice that since in complete lattices, the empty set have also their meet and joins then all complete lattices are bounded. It is also important to note that all finite lattices are complete lattice. However, not all finite meet-semilattices are lattices since they may lack of a top element (i.e. the meet of the empty set is not guaranteed to exist). In fact, a finite meet-semilattice (resp. join-semilattice) is a lattice if and only if it is upper-bounded (resp. lower-bounded). 
\end{note}

\subsection{Morphisms on Posets}

We will often use morphismes (i.e. mappings) between posets in this paper. Definition~\ref{def:morhismproperties} below formulate some properties of morphismes between two posets.

\begin{definition}
\label{def:morhismproperties}
Let $(P, \leq)$ and $(Q, \leq)$  be the two posets. A mapping $f : P \to Q$ is:
\begin{itemize}[noitemsep,topsep=1pt]
\item \empha{order-preserving} or \empha{monotone}: 
$(\forall x,y \in P) \: \: x \leq y  \Rightarrow  f(x) \leq f(y)$.
\item \empha{order-reversing}: 
$(\forall x,y \in P) \: \: x \leq y  \Rightarrow  f(y) \leq f(x)$

\item \empha{an order-embedding}: 
$(\forall x,y \in P) \: \: x \leq y  \Leftrightarrow  f(x) \leq f(y)$
\end{itemize}

If an order-embedding exist from $(P, \leq)$ and $(Q, \leq)$ then poset $(Q, \leq)$ is said to be a \empha{completion} of $(P, \leq)$ or \empha{embeds} $(P, \leq)$. If this order-embedding is surjective (i.e. an \empha{order-isomorphism}) we say that $(P, \leq)$ and $(Q, \leq)$ are \empha{order-isomorphic}.
\end{definition}

 
%


\begin{definition}
\label{def:closureOperator}
A \empha{closure operator} on $(P,\leq)$ is a mapping $\phi : P \to P$ that is: 
\begin{itemize}
\item \empha{monotone:} $(\forall x,y \in P) \: \: x \leq y  \Rightarrow  \phi(x) \leq \phi(y)$,
\item \empha{extensive:} 
\vphantom{vee}$(\forall x \in P) \: \: x \leq \phi(x)$, and 
\item \empha{idempotent:} 
$(\forall x \in P) \: \: \phi(\phi(x)) = \phi(x)$
\end{itemize}
\end{definition}

\begin{note}
The \empha{fixpoints} of a given mapping $f : P \to P$ is the set of elements s.t. $\{p \in P \mid f(p) = p\}$. For idempotent operator, the set of fixpoints is 
$f[P] = \{f(p) \mid p \in P\}$.
\end{note}


\begin{lemma}\label{lemma:downandupclosure}
Let $(P, \leq)$ be a poset, we have $\uparrow$ and $\downarrow$ are closures on $(\wp(P), \subseteq)$. 
\end{lemma}

\begin{mdframed}
\input{proofs/proof3.tex}
\end{mdframed}

\begin{note}
Fr any poset $(P, \leq)$ and since $\uparrow$ and $\downarrow$ are closure operators, we have:
\begin{itemize}[topsep=1pt,noitemsep]
    \item For all $S \in \mathscr{U}(P)$, if $S$ is minimal-handle then $S = \uparrow min(S)$.
    \item For all $S \in \mathscr{O}(P)$, if $S$ is maximal-handle then $S = \downarrow max(S)$.
\end{itemize}
\end{note}

Before finishing the section, we draw the reader attention to the tight link existing between closure operators and closure systems on a complete lattice defined below. 

\begin{definition}
\label{def:closureSystem}
Let $(P, \leq)$ be a complete lattice whose meet is $\bigwedge$. Given a subset $S \subseteq P$, $(S, \leq)$ is said to be a \empha{closure system} or a \empha{meet-structure} on $(P, \leq)$ \emph{iff}:  
\begin{eqnarray*}
(\forall A \subseteq S) \:  \bigwedge A \in S
\end{eqnarray*}

Clearly, $(S, \leq)$ forms a {complete lattice} which infimum is $\bigwedge$.
\end{definition}

In fact (see Theorem~1 in \cite{GanterW99}), if $(P, \leq)$ is a complete lattice and $\phi$ is a closure operator on $(P,\leq)$ then the poset $(\phi[P], \leq)$ of fixpoints of $\phi$ is a \emph{closure system}. Conversely, if $(S, \leq)$ is a closure system on a complete lattice $(P, \leq)$ then the mapping $\phi_S: P \to S, p \mapsto \bigwedge \{s \in S \mid p \leq s\}$ that takes each element $p \in P$ to the smallest element (fixpoint) in $S$ enclosing it is a closure operator with  $\phi_S[P] = S$.

%% file: figures/hasse_diagram_powerset_abc.tex
\begin{tikzpicture}[scale=.7]
\node (0) at (0,-1) {$\emptyset$};
\node (1) at (-2,0.25) {$\{a\}$};
\node (2) at (0,0.25) {$\{b\}$};
\node (3) at (2,0.25) {$\{c\}$};
\node (12) at (-2,1.5) {$\{a,b\}$};
\node (13) at (0,1.5) {$\{a,c\}$};
\node (23) at (2,1.5) {$\{b,c\}$};
\node (123) at (0,3) {$\{a,b,c\}$};
\draw (0) -- (1) -- (12) -- (123) -- (23) -- (3) -- (0);
\draw (0) -- (2) -- (23);
\draw (2) -- (12);
\draw (3) -- (13) -- (1);
\draw (13) -- (123);
\node (t) at (-2,3) {$\mathbf{(1)}$};
\end{tikzpicture}

%% file: figures/hasse_diagram_minimum_minimal.tex
\begin{tikzpicture}[scale=.7]
\node (c) at (0,0) {$\bot$};
\node (a0) at (1,1) {$a_0$};
\node (a1) at (1,2) {$a_1$};
\node (an) at (1,3) 
{$\vdots$};
\node (d) at (0,4) {$\top$};
\node (b0) at (-1,1) {$b_0$};
\node (b1) at (-1,2) {$b_1$};

\draw (c) -- (a0) -- (a1) -- (an) -- (d) -- (b1) -- (b0) -- (c);

\node (t) at (-1,4) {$\mathbf{(2)}$};
\end{tikzpicture}

%% file: figures/hasse_diagram_substringofabandba.tex
\begin{tikzpicture}[scale=.7]
\node (empty) at (0,0) {$``"$};
\node (a) at (-1,1.2) {$``a"$};
\node (b) at (1,1.2) {$``b"$};
\node (ab) at (-1,3) {$``ab"$};
\node (ba) at (1,3) {$``ba"$};
\draw (empty) -- (a) -- (ab) -- (b) -- (empty);
\draw (a) -- (ba) -- (b);
\node (t) at (-1,4) {$\mathbf{(3)}$};
\end{tikzpicture}

%% file: figures/hasse_diagram_impossible_extents.tex
\begin{tikzpicture}[scale=.7]
\node (g12) at (-1.2,0.3) {$\{g_1 ,g_2\}$};
\node (g13) at (1.2,0.3) {$\{g_1 ,g_3\}$};
\node (g123) at (0,1.5) {$\{g_1,g_2,g_3\}$};
\node (g1234) at (0,3) {$\{g_1,g_2,g_3,g_4\}$};
\draw (g12) -- (g123) -- (g13);
\draw (g123) -- (g1234);
\node (t) at (-1,4) {$\mathbf{(4)}$};
\end{tikzpicture}

%% file: proofs/proof1.tex
\begin{proof}
We prove by double inclusion the property $min (\uparrow S)  =  min(S)$:
\begin{itemize}
    \item[\textbf{($\subseteq$)}] Let $x \in min(\uparrow S)$ and suppose that $x \not\in S$. Since $x \in min(\uparrow S)$ we have $x \in \uparrow S$, that is $\exists y \in S$ s.t. $y \leq x$ but $y \neq x$ since $x \not\in S$. Thus, $\exists y \in \uparrow S$ s.t. $y \leq x$ but $y \neq x$. Thus $y \in \downarrow x \cap \uparrow S$ with $y \neq x$ which contradicts the fact that $x \in min(\uparrow S)$ (i.e. $\downarrow x \cap \uparrow S = \{x\}$). We conclude that $x \in S$. Suppose now that $x \not\in min(S)$, that $\exists y \in S$ s.t. $y \leq x$ and $y \neq x$. Hence, $y \in \uparrow S \cap \downarrow x$ which contradicts the fact that $x \in min(\uparrow S)$. Thus $x \in min(S)$. We conclude that $min (\uparrow S) \subseteq min(S)$.
    
    \item[\textbf{($\supseteq$)}] Let $x \in min(S)$, thus $x \in \uparrow S$. Suppose that $x \not\in min(\uparrow S)$ that is $\exists y \in \uparrow S$ such that $y \leq x$ and $y \neq x$. Thus $\exists z \in S$ such that $z \leq y \leq x$ with $z \neq x$. Hence, $z \in \downarrow x \cap S$ with $z \neq x$ which contradicts the fact that $x \in min(S)$. We conclude that $x \in min(\uparrow S)$ or more generally $min (S) \subseteq min(\uparrow S)$.
\end{itemize}

On can follow the same steps to show $max (\downarrow S)  =  max(S)$. 
\end{proof}

%% file: proofs/proof2.tex
\begin{proof}
Let $A \subseteq S^\ell$, we have by definition:
$
(\forall s \in S \:\: \forall a \in A) \:\:\: a \leq s
$,
that is $S \subseteq A^u$. Since $\bigvee A$ is the least upper bound of $A$ and all elements of $S$ are upper bounds of $A$ then:
$
(\forall s \in S) \:\:\: \bigvee A \leq s
$. We conclude that $\bigvee A \in S^\ell$. 
%
%
Same steps can be followed to show the second part of the Lemma.
\end{proof}

%% file: proofs/proof3.tex
\begin{proof}
Let us show that $\uparrow : \wp(P) \to \wp(P)$ is a closure operator on $(\wp(P), \subseteq)$. It is clear that $S \subseteq \uparrow S$ (i.e. $\uparrow$ is \textbf{extensive}) by definition. 
For $S \subseteq T$ in $\wp(P)$, we have if $x \in \uparrow S$, then $\exists y \in S \subseteq T$ such that $y \leq x$. That is, $x \in \uparrow T$. Thus, $(\forall S, T \in \wp(P))\uparrow S \subseteq T \implies \uparrow S \subseteq \uparrow T$ (i.e. $\uparrow$ is \textbf{order-preserving}). 
Let us show that $\uparrow$ is \textbf{idempotent}. It is clear that $\uparrow S \subseteq \uparrow \uparrow S$ since $\uparrow$ is extensive. 
It remains to show that $\uparrow \uparrow S \subseteq \uparrow S$. Let $x \in \uparrow \uparrow S$, that is $\exists y \in \uparrow S$ such that $y \leq x$. That is $\exists z \in S$ such that $z \leq y \leq x$. We conclude that $x \in \uparrow S$.
One can follow the same steps to show that $\downarrow : \wp(P) \to \wp(P)$ is a closure operator on $(\wp(P), \subseteq)$.
\end{proof}

%% file: sections/patternSetups.tex
\newpage
\section{Pattern Setups \label{sec:patternSetups}}
\emph{Formal Concept Analysis (FCA)} were introduced in \cite{wille1982restructuring} as a mathematical framework to analyze and manipulate concepts in databases. 
\emph{FCA} starts by a \emph{formal context} $\mathbb{K} = (\mathcal{G}, \mathcal{M}, \mathcal{I})$ where $\mathcal{G}$ is a set of objects (i.e. \emph{Gegenst{\"a}nde}), $\mathcal{M}$ is a set of attributes (i.e. \emph{Merkmale}) and $\mathcal{I}$ is a binary relation on $\mathcal{G} \times \mathcal{M}$ (i.e. \emph{Incidence relation}). For $(g, m) \in \mathcal{G} \times \mathcal{M}$, $g\,\mathcal{I}\,m$ holds iff $g$ has attribute $m$. \figref{fig:formalcontext} presents an example of a formal context. The basic theorem behinds \emph{FCA} rely on the observation that any formal context can be transformed to a complete lattice called \emph{concept lattice}. We invite the reader to read \cite{GanterW99} book for more details.  

While (basic) FCA gives a tool to analyze datasets in a form of formal context, datasets with more complex attributes (eg. numerical or nominal attributes) needs to be transformed to such a form before any manipulation. Such a transformation is called \emph{conceptual scaling} (i.e. \emph{binarizing}) (\cite{ganter1989conceptual}). 
Yet, even if \emph{conceptual scaling} is a quite general tool, binarizing a dataset with regard to the patterns we want to look for is not always straightforward (\cite{DBLP:conf/cla/BaixeriesKN12}).  

In response to that, a more natural way to handle complex datasets was introduced in~\cite{DBLP:conf/iccs/GanterK01} under the name of \emph{pattern structures}.
Objects in a pattern structure have descriptions (e.g. the equivalent notion to set of attributes in $\mathcal{M}$ in a formal context) with a meet-semilattice operation on them (e.g. equivalent to set intersection in $(\wp(\mathcal{M}), \subseteq)$ in a formal context). This framework proved its usefulness in many data analysis tasks (see \cite{DBLP:conf/rsfdgrc/Kuznetsov09}). 
However, pattern structures demands that the \emph{description space} to be a (upper-bounded) meet-semilattice which is not the case for all description spaces such as sequence of itemsets patterns (\cite{DBLP:conf/icfca/CodocedoBKBN17}). \emph{Pattern setups} were introduced in~\cite{DBLP:conf/cla/LumpeS15} to generalize pattern structures by demanding only a partial order on descriptions. We details in this sections the different notions related to pattern setups.
%
%
%

\begin{figure}[b]
\centering
\begin{minipage}{0.30\linewidth} 
  \centering 
  \begin{tabular}{llll}
  \hline   
  $\:\mathcal{G}\:$  & $a$ & $b$ & $c$\\ 
   \hline
   $g_1$ & $\times$ & $\times$ & $\times$\\ 
   $g_2$ & $\times$ & & \\
   $g_3$ & $\times$ \\ 
   $g_4$ & & $\times$ & $\times$\\
   \hline
 \end{tabular}
 \end{minipage}
\caption{Formal Context $(\mathcal{G}, \mathcal{M}, \mathcal{I})$ with $\mathcal{G} = \{g_i\}_{1\leq i \leq 4}$ and $\mathcal{M} = \{a, b ,c\}$. \label{fig:formalcontext}}
\end{figure}

\begin{definition}
A \empha{description space}; called also \empha{description language}, \empha{pattern space} or \empha{pattern language}; is any \emph{poset} $\underline{\mathcal{D}} := (\mathcal{D}, \sqsubseteq)$. Elements of $\mathcal{D}$ are called \empha{descriptions} or \empha{patterns}. For any $c, d \in \mathcal{D}$, $c \sqsubseteq d$ should be read as ``$c$ is \emph{less restrictive than} $d$" or ``$c$ \empha{subsumes} $d$".
\end{definition}

\begin{definition}
A \empha{pattern setup} is a triple $\mathbb{P}=(\mathcal{G}, \underline{\mathcal{D}}, \delta)$ where $\mathcal{G}$ is a \emph{set} (of objects), $\underline{\mathcal{D}}$ is a description space and $\delta : \mathcal{G} \to \mathcal{D}$ defines a mapping that takes each object $g \in \mathcal{G}$ to its description $\delta(g) \in \mathcal{D}$. 
%
Let $g \in \mathcal{G}$ and $d \in \mathcal{D}$ be  an object and a description, respectively. We say that object $g$ \empha{realizes} description $d$ or description $d$ \empha{hold for} or \empha{cover} object $g$ iff $d \sqsubseteq \delta(g)$.
\end{definition}


\begin{example}\label{example:patternsetup:presentation}
Consider the pattern setup $\mathbb{P} = (\mathcal{G}, \underline{\mathcal{D}}, \delta)$ in \figref{fig:exampleSequences}. We have $\mathcal{G} = \{g_i\}_{1 \leq i \leq 4}$. 
The description space is the set of \emph{\underline{nonempty} words} on the alphabet $\{a,b,c\}$ (i.e. $\{a,b,c\}^+$) ordered by the relationship \emph{``is substring of"} $\sqsubseteq$. The mapping $\delta$ associates to each objects in $\mathcal{G}$ its word in the description space. For instance $\delta(g_1) = ``cab"$. The diagram in the center of \figref{fig:exampleSequences} depicts the Hasse Diagram of the poset $(\downarrow \delta[\mathcal{G}], \sqsubseteq)$ with $\delta[\mathcal{G}] = \{``\mathbf{cab}",``\mathbf{cbba}",``\mathbf{a}",``\mathbf{bbc}"\}$. In other words, it depicts the set of descriptions $d \in \mathcal{D}$ holding for at least one object in $\mathcal{G}$. 
It is clear that the description $``ca"$ holds for $g_1$ since $``ca" \sqsubseteq ``cab"$. However, description $``cb"$ does not hold for $g_1$ since $``cb"$ is not a substring of $``cab"$. 
More generally, descriptions holding for $g_i$ is the principal filter of $\delta(g_i)$ (i.e. $\downarrow \delta(g_i)$). For instance, the set of descriptions holding for $g_1$ is given by $\downarrow \delta(g_1) = \{``a", ``ca", ``ab", \mathbf{``cab"}\}$. 

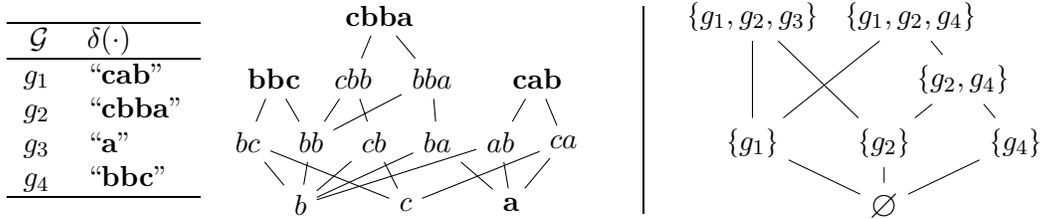
\begin{figure}[t]
\begin{minipage}{1\linewidth}
  \centering 
  \begin{minipage}{0.2\linewidth} 
  \centering 
  \begin{tabular}{ll}
  \hline   
  $\:\mathcal{G}\:$  & $\delta(\cdot)$\\ 
   \hline
   $g_1$ & $``\mathbf{cab}"$\\ 
   $g_2$ & $``\mathbf{cbba}"$\\
   $g_3$ & $``\mathbf{a}"$\\ 
   $g_4$ & $``\mathbf{bbc}"$\\
   \hline
 \end{tabular}
 \end{minipage}%
 \hfill
 \begin{minipage}{0.35\linewidth} 
  \centering 
  \scalebox{1}{\input{figures/hassediagramSequenceDescriptions.tex}}
 \end{minipage}%
 \hfill
 \begin{minipage}{0.03\linewidth} 
  \centering 
  \begin{tabular}{c|}
  \\\\\\\\\\\\
 \end{tabular}
 \end{minipage}
 \hfill
 \begin{minipage}{0.38\linewidth} 
  \centering 
  \scalebox{1}{\input{figures/hassediagramSequenceExample.tex}}
 \end{minipage}%
\end{minipage}
\caption{
The table (left) represents the mapping function $\delta$ of the pattern setup considered in running example \ref{example:patternsetup:presentation}. 
The diagram (center)
represents the set of non empty substrings in $\{a,b,c\}^+$ holding for at least one object in $\mathcal{G}$.
The diagram (right) represents the poset of definable sets $(\mathbb{P}_{ext}, \subseteq)$\label{fig:exampleSequences}.}
\end{figure}
\end{example}

\begin{example}
\label{example:patternstructure:presentation}
Consider the \emph{pattern setup} $\mathbb{P} = (\mathcal{G}, \underline{\mathcal{D}}, \delta)$  presented in Fig.~\ref{fig:exampleItemsets} (left). The set of objects is $\mathcal{G} = \{g_i\}_{1 \leq i \leq 4}$ and the description space $\underline{\mathcal{D}}$ is the powerset ordered by set inclusion $(\wp(\mathcal{M}), \subseteq)$ (i.e. itemsets) with $\mathcal{M} = \{a, b, c\}$. Again descriptions holding for $g_4$ are all itemsets included in $\delta(g_4) = \{b,c\}$ (i.e. $\downarrow \delta(g_4) = \{\emptyset, \{b\}, \{c\}, \{b,c\}\}$).
\end{example}

\subsection{On Extent and Cover Operators}
We have seen that the relation \empha{``realizes"} build a binary relation between objects and descriptions. Based on this binary relation, two key operators, namely \empha{extent} and \empha{cover} are derived (see Definition \ref{def:extentoperator} and \ref{def:coverOperator}).
\begin{definition}
\label{def:extentoperator}
The \empha{extent operator}, denoted by $ext$, is the operator that takes each description $d \in \mathcal{D}$ to the subset of objects in $\mathcal{G}$ \emph{realizing} it:
\begin{eqnarray*}
\lfunction{ext}{\mathcal{D}}{\wp(\mathcal{G})}{d}{\{g \in \mathcal{G} \mid d \sqsubseteq \delta(g)\}} 
\end{eqnarray*}
The size of $ext(d)$ is called the \empha{support} of $d$, i.e. $support : d \mapsto |ext(d)|$.
\end{definition}

\begin{note}
Please note that for any $S \subseteq \mathcal{D}$, we denote $ext[S] = \{ext(d) \mid d \in S\}$.
\end{note}


\begin{definition}
\label{def:coverOperator}
The \empha{cover operator}, denoted by $cov$, takes each subset of objects $A \subseteq \mathcal{G}$ to the set of common descriptions in $\mathcal{D}$ \emph{covering} all of them:
\begin{eqnarray*}
\lfunction{cov}{\wp(\mathcal{G})}{\wp(\mathcal{\mathcal{D}})}{A}{\delta[A]^\ell = \{d \in \mathcal{D} \mid (\forall g \in A) \: d \sqsubseteq \delta(g) \}}
\end{eqnarray*}
\end{definition}

\begin{example}
Consider the pattern setup presented in Example \ref{example:patternsetup:presentation}. We have:
$ext(``bb") = \{g_2, g_4\}$ and $cov(\{g_2, g_4\}) = \downarrow \delta(g_2) \cap  \downarrow \delta(g_4) = \{``b", ``bb", ``c"\}$. 
\end{example}

\begin{definition}
A subset $A \subseteq \mathcal{G}$ is said to be:
\begin{itemize}
\item \empha{Definable}, \empha{Separable} or an \empha{Extent} if $(\exists d \in A)$ $A = ext(d)$.
\item \empha{Coverable} if $cov(A) \neq \emptyset$. 
\end{itemize}

The set of \empha{definable sets} is then given by: 
\begin{eqnarray*}
\mathbb{P}_{ext} = ext[\mathcal{D}] = \{ext(d) \mid d \in \mathcal{D}\}
\end{eqnarray*}
\end{definition}

\begin{note}
One should note that poset $(\mathbb{P}_{ext}, \subseteq)$ does form a subposet of $(\wp(\mathcal{G}), \subseteq)$, that is definable sets are naturally ordered by $\subseteq$.
The set of \emph{coverable sets} is naturally given by $\downarrow \mathbb{P}_{ext}$. 
In other words, any subset of a \emph{coverable set} is \emph{coverable}. Conversely, any superset of a \emph{non coverable set} is a \emph{non coverable set}.
\end{note}

\begin{example}
\label{example:definableandcoverable}
The poset of definable sets $(\mathbb{P}_{ext}, \subseteq)$ associated to the pattern setup presented in Example \ref{example:patternsetup:presentation} is depicted in \figref{fig:exampleSequences} (right). For example, it is clear that $\{g_2,g_4\}$ is definable since $ext(``bb") = \{g_2,g_4\}$. However, there is no description which extent is exactly $\{g_1,g_2\}$, hence $\{g_1,g_2\}$ will be said \empha{non-definable}. Still, $\{g_1,g_2\}$ is coverable since $g_1$ and $g_2$ share at least one common descritpion (i.e. $cov(\{g_1,g_2\} = \{``a",``b",``c"\} \neq \emptyset$). One should note also that $\{g_3,g_4\}$ is \empha{non-coverable} since they share no common symbol and the empty string is excluded from the pattern space. 
\end{example}

An important property arising directly from the definition of both $ext$ and $cov$ is given in Proposition~\ref{prop:orderreversingmappints}. It tells that: on the one hand, the more restrictive is a description, the less it covers objects in the database. On the other hand, the bigger is a set of objects the less they share descriptions in common. 

\begin{proposition}\label{prop:orderreversingmappints}
Operators $ext$ and $cov$ are \emph{order-reversing}:
\begin{eqnarray*}
(\forall c,d \in \mathcal{D})\:\:c \sqsubseteq d \Rightarrow ext(d) \subseteq ext(c)  & \phantom{:::} & (\forall A,B \subseteq \mathcal{G})\:\:A \subseteq B \Rightarrow cov(B) \subseteq cov(A)
\end{eqnarray*}
\end{proposition}

\begin{mdframed}
\input{proofs/proof5.tex}  
\end{mdframed}

\begin{note}
We have seen in Proposition \ref{prop:orderreversingmappints} that mappings $ext$ and $cov$ are order reversing. Hence, one could think that $(ext, cov)$ may form some Galois connection\footnote{A Galois connection between two posets $(P, \leq)$ and $(Q, \sqsubseteq)$ is a pair $(f,g)$ with $f : P \to Q$ and $g : Q \to P$ are order-reversing and both operators $f\circ g$ and $g \circ f$ are extensive.}. However, it is not the case since $ext$ associate an extent to one description while $cov$ outputs a set of descriptions rather than one. In other words, these two mappings are not compatibles. Yet, we will see in next section that $ext$ will be involved into a Galois connection when the considered pattern setup verifies additional properties (i.e. the pattern setup is a pattern structure \cite{DBLP:conf/iccs/GanterK01}). Mapping $cov$ will also be involved in another Galois connection in Section~\ref{sec:completions}.
\end{note}

\begin{definition}
For $c,d \in \mathcal{D}$, the \empha{pattern implication} $c \to d$ holds if $ext(c) \subseteq ext(d)$. That is, every object realizing $c$ realizes $d$. Dually, for $A,B \subseteq \mathcal{G}$, the \empha{object implication} $A \to B$ holds if $cov(A) \subseteq cov(B)$. That is, every description covering all object in $A$ covers also all objects in $B$.

We say that descriptions $c, d \in \mathcal{D}$ are \empha{equivalent} if $c \to d$ and $d \to c$ and we have $ext(c) = ext(d)$. Dual definition can be given for the \emph{equivalence between object sets}.
\end{definition}

\begin{note}
Please note that if $d \sqsubseteq c$ and since $ext$ is an \emph{order-reversing mapping}, we have $c \to d$. Regarding this observation, there is two types of implications between descriptions: (i) implications deduced directly from $\sqsubseteq$ and (ii) implications that are dependent on the pattern setup. While the former implications are intrinsic to the description space, the latter are more important since they are those who enclose the knowledge hidden in the pattern setup.    
\end{note}

\begin{example}
In the pattern setup presented in Example \ref{example:patternsetup:presentation} and \figref{fig:exampleSequences}, we have $ext(``bb") = \{g_2,g_4\}$ and $ext(``c") = \{g_1,g_2,g_4\}$. Hence, we have $``bb" \rightarrow ``c"$ or in other words in every string containing $``bb"$ in the pattern setup contains also $``c"$.
\end{example}

Proposition \ref{prop:compositionbetweenextandcov} gives characterizations of the set $ext[cov(A)]$ and $cov(ext(d))$ for $A \subseteq \mathcal{G}$ and $d \in \mathcal{D}$. This proposition will be useful later in this paper. 
\begin{proposition}\label{prop:compositionbetweenextandcov}
For $A \subseteq \mathcal{G}$ and $d \in \mathcal{D}$:
\begin{eqnarray*}
ext[cov(A)] & = & \{E \in \mathbb{P}_{ext} \mid A \subseteq E\} = \uparrow A \cap \mathbb{P}_{ext} \\
cov(ext(d)) & = & \{c \in \mathcal{D} \mid ext(d) \subseteq ext(c)\} = \{c \in \mathcal{D} \mid d \to c\}
\end{eqnarray*}

\end{proposition}

\begin{mdframed}
\input{proofs/proof6.tex}
\end{mdframed}

\begin{example}
Consider the pattern setup presented in Example \ref{example:patternsetup:presentation} and its associated $(\mathbb{P}_{ext}, \subseteq)$ depicted in \figref{fig:exampleSequences} (right). We have:
$ext(``bb") = \{g_2, g_4\}$ and $cov(\{g_2, g_4\}) = \{``b", ``bb", ``c"\}$. 
Hence:
\begin{itemize}[topsep=1pt,noitemsep]
\item $ext[cov(\{g_2, g_4\})] = \{ext(``b"),ext(``bb"),ext(``c")\} = \{\{g_1,g_2,g_4\},\{g_2,g_4\}\}$.

\item $cov(ext(``bb")] = cov(\{g_2,g_4\}) = \{``b", ``bb", ``c"\}$.
\end{itemize}
\end{example}

\subsection{A Minimal Representation of a Pattern Setup}

An important notion analogous to what is called \emph{representation context} in pattern structures (see~\cite{DBLP:conf/iccs/GanterK01,DBLP:conf/icfca/BuzmakovKN15}) is introduced in Theorem \ref{thm:representational}. Technically, such a representation does not provide a practical way to explore definable sets of an arbitrary pattern setups, but helps to simulate definable sets search space of a pattern setup independently from the description space. 
Before introducing the Theorem, we present Proposition \ref{prop:setsystemisaextentset}. 

\begin{proposition}\label{prop:setsystemisaextentset}
Let $\mathcal{G}$ be a non empty finite set and let $S \subseteq \wp(\mathcal{G})$, we have 
\begin{eqnarray*}
\exists \mathbb{P} \emph{ a pattern setup such that } S = \mathbb{P}_{ext} & \Longleftrightarrow & (\forall g \in \mathcal{G}) \text{ } \bigcap \: (\uparrow g \cap S) \in S
\end{eqnarray*}
\end{proposition}

\begin{mdframed}
\input{proofs/proof8.tex}

\end{mdframed}

\begin{example}
Proposition \ref{prop:setsystemisaextentset} tells that not all families of subsets of $\mathcal{G}$ could be seen as a set of extents of some pattern setup. Consider the poset depicted in \figref{fig:hassediagramSimple} \textbf{(4)} where $\mathcal{G} = \{g_1,g_2,g_3,g_4\}$ and $S = \{\{g_1,g_2\},\{g_1,g_3\},\{g_1,g_2,g_3\},\{g_1,g_2,g_3,g_4\}\}$ can never be seen as a $\mathbb{P}_{ext}$ for some pattern setup $\mathbb{P}$. Indeed, $\uparrow g_1 \cap S = \{\{g_1,g_2\},\{g_1,g_3\}\}$ whose intersection is not in $S$.
\end{example}




\begin{theorem}
\label{thm:representational}
For any pattern setup $\mathbb{P}$, 
the pattern setup $\mathbb{R}(\mathbb{P})$ given by
\begin{eqnarray*}
\mathbb{R}(\mathbb{P}) = \left( \mathcal{G}, (\mathbb{P}_{ext}, \supseteq), g \mapsto \bigcap (\uparrow g \cap \mathbb{P}_{ext}) \right)
\end{eqnarray*}
is called the \empha{minimal representation} of $\mathbb{P}$ and we have $\mathbb{R}(\mathbb{P})_{ext} = \mathbb{P}_{ext}$.
\end{theorem}

\begin{mdframed}
\input{proofs/proof7.tex}
\end{mdframed}

%% file: figures/hassediagramSequenceDescriptions.tex
\begin{tikzpicture}[scale=.7]
\node (b) at (-2,0) {$b$};
\node (c) at (0,0) {$c$};
\node (a) at (2,0) {$\mathbf{a}$};

\node (bc) at (-3,1.2) {$bc$};
\node (bb) at (-1.8,1.2) {$bb$};
\node (cb) at (-0.6,1.2) {$cb$};
\node (ba) at (0.6,1.2) {$ba$};
\node (ab) at (1.8,1.2) {$ab$};
\node (ca) at (3,1.2) {$ca$};

\node (bbc) at (-2.5,2.4) {$\mathbf{bbc}$};
\node (cbb) at (-1,2.4) {$cbb$};
\node (bba) at (0.5,2.4) {$bba$};
\node (cab) at (2.5,2.4) {$\mathbf{cab}$};

\node (cbba) at (-0.5,3.6) {$\mathbf{cbba}$};

\draw (b) -- (bb) -- (bbc);
\draw (b) -- (bc) -- (bbc);
\draw (b) -- (cb) -- (cbb) -- (cbba) -- (bba) -- (bb) -- (cbb);
\draw (b) -- (ba) -- (bba);
\draw (b) -- (ab) -- (cab) -- (ca) -- (a) -- (ba);
\draw (a) -- (ab);
\draw (c) -- (ca);
\draw (c) -- (cb);
\draw (c) -- (bc);
\end{tikzpicture}

%% file: figures/hassediagramSequenceExample.tex
\begin{tikzpicture}[scale=.7]
\node (0) at (-0.5,-1.2) {$\emptyset$};
\node (1) at (-3,0) {$\{g_1\}$};
\node (2) at (-0.5,0) {$\{g_2\}$};
\node (4) at (2,0) {$\{g_4\}$};
\node (123) at (-3,2.4) {$\{g_1,g_2,g_3\}$};
\node (24) at (1,1.2) {$\{g_2,g_4\}$};
\node (124) at (0,2.4) {$\{g_1,g_2,g_4\}$};
\draw (0) -- (1) -- (123);
\draw (124) -- (1);
\draw (0) -- (2) -- (24) -- (124);
\draw (2) -- (123);
\draw (0) -- (4) -- (24);
\end{tikzpicture}

%% file: proofs/proof5.tex
\begin{proof}
We have:
\begin{enumerate}
\item Let $A \subseteq B \subseteq \mathcal{G}$, let $d \in cov(B)$, thus $(\forall g \in B) \: d \sqsubseteq \delta(g)$. Since $A \subseteq B$, we conclude that $(\forall g \in A) \: d \sqsubseteq \delta(g)$ that is $d \in cov(A)$. Thus $cov(B) \subseteq cov(A)$.
\item Let $c, d \in \mathcal{D}$ such that $c \sqsubseteq d$. Let $g \in ext(d)$, that is $d \sqsubseteq \delta(g)$ thus $c \sqsubseteq \delta(g)$; that is $ g \in ext(c)$. We conclude that $ext(d) \subseteq ext(c)$.
\end{enumerate}
This concludes the proof.
\end{proof}

%% file: proofs/proof6.tex
\begin{proof}
We show the two equations separately:

\textbf{(1)} Let $E \subseteq \mathcal{G}$, we have:
\begin{eqnarray*}
E \in ext[cov(A)] 
& \Leftrightarrow (\exists d \in cov(A)) \: E = ext(d) 
& \Leftrightarrow (\exists d \in \mathcal{D} \forall g \in A) \: d \sqsubseteq \delta(g) \\
& \Leftrightarrow  (\exists d \in \mathcal{D}) \: A \subseteq ext(d) = E 
& \Leftrightarrow E \in \mathbb{P}_{ext}\,\cap \uparrow A
\end{eqnarray*}

We conclude that $ext[cov(A)] = \{E \in \mathbb{P}_{ext} \mid A \subseteq E\} = \uparrow A \cap \mathbb{P}_{ext}$

\smallskip

\textbf{(2)} Let $c \in \mathcal{D}$, we have:
\begin{eqnarray*}
c \in cov(ext(d)) 
 \Leftrightarrow(\forall g \in ext(d)) \: c \sqsubseteq  \delta(g)
\Leftrightarrow (\forall g \in ext(d)) \: g \in ext(c) \\
\Leftrightarrow ext(d) \subseteq ext(c)
\end{eqnarray*}

Thus $cov(ext(d)) = \{c \in \mathcal{D} \mid ext(d) \subseteq ext(c)\} = \{c \in \mathcal{D} \mid c \to d\}$. 
\end{proof}

%% file: proofs/proof8.tex
\begin{proof}
Before showing the equivalence let us prove the following property:
\begin{eqnarray}\label{eq:proof316eq1}
(\forall g \in \mathcal{G})  & \bigcap \left(\uparrow g \cap \mathbb{P}_{ext}\right) = ext(\delta(g)) \in \mathbb{P}_{ext}
\end{eqnarray}

Recall that $ext[cov(\{g\})] = \uparrow g \cap \mathbb{P}_{ext}$ (See proposition \ref{prop:compositionbetweenextandcov}). 
Let $g \in \mathcal{G}$, we have  $\delta(g) \in cov(\{g\})$, thus $ext(\delta(g)) \in ext[cov(\{g\})]$. Let us show that $ext(\delta(g))$ is a lower bound of $ext[cov(\{g\})]$. We have:
$ cov(\{g\}) = \{d \in \mathcal{D} \mid d \sqsubseteq \delta(g)\} $. Thus, $\forall d \in cov(\{g\}) : d \sqsubseteq \delta(g)$. Since $ext$ is an order reversing operator, we obtain: $\forall A \in ext[cov(\{g\})] : ext(\delta(g)) \subseteq A$. Thus, $ext(\delta(g))$ is the smallest element of $ext[cov(\{g\})]$. That is $\bigcap \left(\uparrow g \cap \mathbb{P}_{ext}\right) = ext(\delta(g))$. 

We show now the two implications independently: 
\begin{itemize}
    \item[($\Rightarrow$)] Let $S = \mathbb{P}_{ext}$ for some pattern setup $\mathbb{P}$. Using eq.~(\ref{eq:proof316eq1}), $\bigcap \left(\uparrow g \cap S\right) \in S$.
    
    \item[($\Leftarrow$)] Let $S \subseteq \wp(\mathcal{G})$ for which $\forall g \in \mathcal{G}$ we have $\bigcap(\uparrow g \cap S) \in S$. Let us now define the following pattern setup:
    \begin{eqnarray*}
    \mathbb{P} = \left(\mathcal{G}, (S, \supseteq), \delta : g \mapsto \bigcap(\uparrow g \cap S)\right)
    \end{eqnarray*}
    Let $A \in S$ be a description, we have: $ext(A) = \{g \in \mathcal{G} \mid  \bigcap(\uparrow g \cap S) \subseteq A\}$. Let us show that $ext(A) = A$ by showing double inclusion: \textbf{(1)} Let $g \in A$, thus $A \in (\uparrow g \cap S)$. It follows that $\bigcap (\uparrow g \cap S) \subseteq A$. We conclude that $g \in ext(A)$. Therefor $A \subseteq ext(A)$. \textbf{(2)} Let $g \in ext(A)$, thus $\bigcap (\uparrow g \cap S) \subseteq A$. Since $\forall B \in \uparrow g \cap S$ we have $g \in B$, we have $g \in \bigcap (\uparrow g \cap S)$, that is $g \in A$. We conclude that $ext(A) \subseteq A$. Both inclusion leads us to have $(\forall A \in S) \: ext(A) = A$. We conclude that $ext[S] = S$. In other words, $\mathbb{P}_{ext} = S$.
\end{itemize}

This concludes the proof.
\end{proof}

%% file: proofs/proof7.tex
\begin{proof}
Theorem \ref{thm:representational} is a corollary of Proposition~\ref{prop:setsystemisaextentset}. Indeed, the pattern setup $\mathbb{R}(\mathbb{P})$ is the same as the one built in the proof of Proposition \ref{prop:setsystemisaextentset} $(\Leftarrow)$ since $\mathbb{P}_{ext}$ is a set system verifying the property $(\forall g \in \mathcal{G}) \: \bigcap (\uparrow g \cap \mathbb{P}_{ext}) \in \mathbb{P}_{ext}$ (i.e. implication ($\Rightarrow$)). 
We have $\mathbb{R}(\mathbb{P})_{ext} = \mathbb{P}_{ext}$.
Moreover, this representation is said to be minimal since any proper subposet of  $(\mathbb{P}_{ext}, \supseteq)$ will drop at least one definable set.
\end{proof}

%% file: sections/patternStructures.tex
\section{Pattern Structures\label{sec:patternStructure}}
Pattern structures were introduced in \cite{DBLP:conf/iccs/GanterK01}. They require that every set of objects has a greatest common description (least general generalization). A formal definition is given in Definition \ref{def:patternstructure}. Pattern structures provide a very strong tool to formalize a large class of pattern languages \citep{DBLP:conf/rsfdgrc/Kuznetsov09}. For instance, pattern setups over the language of itemsets \citep{GanterW99}, intervals \citep{DBLP:conf/ijcai/KaytoueKN11}, convex polygons \citep{DBLP:conf/ijcai/BelfodilKRK17}, sequence sets \citep{DBLP:journals/ijgs/BuzmakovEJKNR16}\footnote{Sequence and graphs patterns will be discussed in the next section, since they do not induce pattern structures directly, but the sets of incomparable patterns do.}, and graph sets\textsuperscript{2}\citep{DBLP:conf/pkdd/Kuznetsov99} are all pattern structures. 

\begin{definition}
\label{def:patternstructure}
A pattern setup $\mathbb{P}=(\mathcal{G}, \underline{\mathcal{D}}, \delta)$ is said to be a \empha{pattern structure} iff every subset of objects has a greatest common description. Formally:
\begin{eqnarray*}
(\forall S \subseteq \delta[\mathcal{G}]) \:\: S \text{ does have  a meet } \bigsqcap S
\end{eqnarray*}
\end{definition}

\begin{example}
\label{example:patternstructure:presentation2}
The pattern setup presented in Example~\ref{example:patternstructure:presentation} and \figref{fig:exampleItemsets} is a pattern structure. Indeed, since the description space is the powerset lattice $(\wp(\{a,b,c\}), \subseteq)$ (i.e a complete lattice) then every subsets $S \subseteq \delta[G] \subseteq \mathcal{D}$ does have a meet which is the set intersection $\bigcap S$. 

However, the pattern setup $\mathbb{P}$ presented in Example \ref{example:patternsetup:presentation} and \figref{fig:exampleSequences} is not a pattern structure. Indeed, the set of common descriptions $cov(\{g_2,g_4\}) = \{``b", ``bb", ``c"\}$ does not have a maximum (i.e. $\{\delta(g_2),\delta(g_4) \}$ does not have a meet) since $\{``b", ``bb", ``c"\}$ has two maximal elements.

\end{example}

One can define a new operator, namely the \empha{intent}, in a pattern structure thanks to the existence of the meet.

\begin{definition}
\label{def:extentIntentOperators}
The \empha{intent operator}, denoted by $int$, is the operator that takes each subset of objects $A \subseteq \mathcal{G}$ to the greatest common description in $\mathcal{D}$ \emph{covering} them (i.e. the maximum of $cov(A)$). Formally: 
\begin{eqnarray*}
\lfunction{int}{\wp(\mathcal{G})}{\mathcal{D}}{A}{inf\,\delta[A] =\bigsqcap \delta[A]}
\end{eqnarray*}
\end{definition}

\begin{note}
In a pattern structure, the pair of operators $(ext, int)$ forms a \emph{Galois connection} between posets $(\wp(\mathcal{G}), \subseteq)$ and $(\mathcal{D}, \sqsubseteq)$. Thus, $ext \circ int$ and $int \circ ext$ form \emph{closure operators} (cf. Proposition 8 in \cite{GanterW99} book) on the two posets respectively. Thanks to this Galois Connection, one can define a complete lattice based on the the closed elements. 
\end{note}

\begin{definition}
Let $\mathbb{P}=(\mathcal{G}, \underline{\mathcal{D}}, \delta)$ be a \emph{pattern structure}. The \empha{(pattern) concept lattice} associated to $\mathbb{P}$ is the \emph{complete lattice} denoted by $\underline{\mathfrak{B}}(\mathbb{P}) = (\mathfrak{B}(\mathbb{P}), \leq)$. Elements of $\mathfrak{B}(\mathbb{P})$ are called \empha{(pattern) concepts} and are given by:
\begin{eqnarray*}
(A, d) \in \wp(\mathcal{G}) \times \mathcal{D} & \emph{ s.t. } & A = ext(d) \emph{ and } d = int(A)
\end{eqnarray*}
The concepts are ordered by $\leq$ as follows:
$(A_1,d_1) \leq (A_2, d_2) \Leftrightarrow A_1 \subseteq A_2 \Leftrightarrow d_2 \sqsubseteq d_1$.
\end{definition}

\begin{note}
Two complete lattices isomorphic to the concept lattice can be derived: 
\begin{enumerate}[topsep=1pt]
\item The poset of \emph{definable sets} $(\mathbb{P}_{ext}, \subseteq)$ which on a \emph{$\bigcap$-structures} (i.e. \emph{Moore family}, \emph{closure system}) in the powerset lattice $(\wp(\mathcal{G}), \subseteq)$. Note that \emph{definable sets} are the fixpoints of the closure operator $ext \circ int$.
\item The poset of \emph{closed patterns} $(\mathcal{D}_\delta, \sqsupseteq)$ with $\mathcal{D}_\delta = int[\wp(\mathcal{G})] = \{\bigsqcap \delta[A] \mid A \subseteq \mathcal{G}\}$ is a complete lattice. Elements of $\mathcal{D}_\delta$ are called \emph{closed patterns} since they are fixpoints of the closure operator $int \circ ext$.
\end{enumerate}
\end{note}

Another important remark about the closure operator $ext \circ int$ is that it takes to a subset of object $A \subseteq \mathcal{G}$ the smallest definable set $E \in \mathbb{P}_{ext}$ enclosing it. Formally:
\begin{proposition}\label{prop:extintclosure}
Let $\mathbb{P} = (\mathcal{G}, (\mathcal{D}, \sqsubseteq), \delta)$ be a pattern structure, we have:
\begin{eqnarray*}
\lfunction{ext \circ int}{\wp(\mathcal{G})}{\wp(\mathcal{G})}{A}{\bigcap \big\{E \in \mathbb{P}_{ext} \mid A \subseteq E\big\} = \bigcap \left(\uparrow A \cap \mathbb{P}_{ext}\right)}
\end{eqnarray*}
\end{proposition}

\begin{mdframed}
\input{proofs/proof9.tex}
\end{mdframed}

\begin{example}\label{example:patternstructure:presentation3}
Consider again the \emph{pattern structure} $\mathbb{P} = (\mathcal{G}, (\wp(\{a,b,c\}),\subseteq), \delta)$ presented in Fig.~\ref{fig:exampleItemsets}. 
Since the meet is the set intersection we have: 
\begin{eqnarray*}
int(\{g_1, g_4\}) = \delta(g_1) \cap \delta(g_4) = \{a, b , c\} \cap \{b ,c\} = \{b ,c\}
\end{eqnarray*}

The concept lattice $\underline{\mathcal{B}}(\mathbb{P})$ is  depicted in \figref{fig:exampleItemsets} (right). 
One should note the set of definable sets $(\mathbb{P}_{ext}, \subseteq)$ can be deduced directly from the concept lattice by taking extents of the pattern concepts.
It is important to highlight the fact that that pattern structure $\mathbb{P}$ is derived from the formal context $\mathbb{K} = (\mathcal{G}, \mathcal{M}, \mathcal{I})$ presented in \figref{fig:formalcontext} where $\delta$ is given by $\delta : g \mapsto \{m \in \mathcal{M} \mid g\, \mathcal{I} \, m\}$.
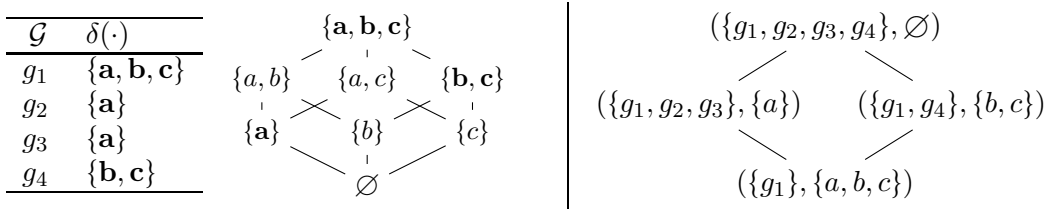
\begin{figure}[t]
\centering
  \hfill
  \begin{minipage}{0.20\linewidth} 
  \centering 
  \begin{tabular}{ll}
  \hline   
  $\:\mathcal{G}\:$  & $\delta(\cdot)$\\ 
   \hline
   $g_1$ & $\mathbf{\{a,b,c\}}$\\ 
   $g_2$ & $\mathbf{\{a\}}$\\
   $g_3$ & $\mathbf{\{a\}}$\\ 
   $g_4$ & $\mathbf{\{b, c\}}$\\
   \hline
 \end{tabular}
 \end{minipage}
 \hfill
 \begin{minipage}{0.28\linewidth}
\scalebox{1}{\input{figures/hasse_diagram_powerset_abc_without_tag.tex}}
\end{minipage}
\hfill
 \begin{minipage}{0.03\linewidth} 
  \centering 
  \begin{tabular}{c|}
  \\\\\\\\\\\\
 \end{tabular}
 \end{minipage}
 \hfill
 \begin{minipage}{0.46\linewidth} 
  \centering 
  \scalebox{1}{\input{figures/hassediagramitemsets.tex}}
 \end{minipage}%
 \hfill
\caption{The table (left) represents the pattern setup $\mathbb{P} = \big(\mathcal{G}, \underline{\mathcal{D}}, \delta\big)$ with $\mathcal{G} = \{g_i\}_{1 \leq i \leq 4}$, $\underline{\mathcal{D}} = (\wp(\{a,b,c\}), \subseteq)$ depicted by the Hasse Diagram (center) and $\delta$ maps an object to its itemset. The diagram (right) depicts the concept lattice $\underline{\mathcal{B}}(\mathbb{P})$. 
\label{fig:exampleItemsets}}
\end{figure}
\end{example}

Definition \ref{def:patternstructure} follows \cite{DBLP:conf/cla/LumpeS15}. The original equivalent definition of the pattern structure  \citep{DBLP:conf/iccs/GanterK01} requires that the description space must be a complete lattice. Theorem \ref{thm:upperboundedmeetsemiolattice} builds a bridge between meet-semilattices and pattern structures over finite set of objects or more generally between pattern structures and complete lattices.

\newpage
\begin{theorem}\label{thm:upperboundedmeetsemiolattice}
Let $\underline{\mathcal{D}} = (\mathcal{D}, \sqsubseteq)$ be a poset, the following properties are equivalent:
\begin{itemize}[topsep=1pt]
\item For any finite set $\mathcal{G} \neq \emptyset$ and any $\delta \in \mathcal{D}^\mathcal{G}$, $(\mathcal{G},\underline{\mathcal{D}}, \delta)$ is a pattern structure (where $\mathcal{D}^\mathcal{G}$ denotes the set of all mappings $\delta : \mathcal{G} \to \mathcal{D}$).  
\item $\underline{\mathcal{D}}$ is a an upper-bounded \emph{meet-semilattice} (i.e. $\emptyset$ has also a \emph{meet}).
\end{itemize}

For an arbitrary set $\mathcal{G}$, the following properties are equivalent:
\begin{itemize}[topsep=1pt]
\item For any set $\mathcal{G} \neq \emptyset$ and any $\delta \in \mathcal{D}^\mathcal{G}$, $(\mathcal{G},\underline{\mathcal{D}}, \delta)$ is a pattern structure.  
\item $\underline{\mathcal{D}}$ is a \emph{complete lattice}.
\end{itemize}

\end{theorem}

\begin{mdframed}
\input{proofs/proof10.tex}
\end{mdframed}


\bigskip

The state-of-the-art abounds with examples of descriptions spaces that are complete lattices that someone can use to build pattern structures:
\begin{itemize}
    \item \emph{Itemset pattern structure \citep{GanterW99}.} The description space is the Boolean lattice $(\wp(M), \subseteq)$ where $M$ is a non empty finite set of attributes.
    \item \emph{Interval pattern structure \citep{DBLP:conf/ijcai/KaytoueKN11}.} The description space is the complete lattice $(\mathcal{C}(\mathbb{R})^m, \sqsupseteq)$\footnote{$\mathcal{C}(E)$ is the set of convex subsetes of $E$. The set $\mathcal{C}(\mathbb{R})$ is the set of all possible intervals of $\mathbb{R}$, $\mathcal{C}(\mathbb{R})^m$ is then the set of all axis-parallel $m$-dimensional hyperrectangles.}, where $\mathcal{C}(\mathbb{R})^m$ represents the set of all possible axis-parallel $m$-dimensional hyperrectangles in $\mathbb{R}^m$ ($m$ is the number of attributes) and $\sqsubseteq$ represents the hyperrectangle inclusion.
    \item \emph{Convex sets pattern structure \citep{DBLP:conf/ijcai/BelfodilKRK17}.} The description space is the complete lattice of all convex sets in $\mathbb{R}^m$ ordered by inclusion $(\mathcal{C}(\mathbb{R}^m), \supseteq)$.
    \item \emph{Partition pattern structure \citep{DBLP:conf/ecai/CodocedoN14}.} The description space is the complete lattice of all partitions $(\mathscr{B}(E), \sqsubseteq)$ of some finite set $E$. The order $\sqsubseteq$ is \emph{finer-than} order relation between partitions. That is for $P_1, P_2 \in \mathscr{B}(E)$ two partitions, $P_1 \sqsubseteq P_2$ if and only $(\forall E_1 \in P_1 \: \exists E_2 \in P_2) \: \: E_1 \subseteq E_2$.  
\end{itemize}

\begin{note}
Before finishing this section, let us highlight some key differences between arbitrary pattern setups and pattern structures. It is clear that the main difference is the fact that the greatest common description does not necessarily exist for any subsets of objects in an arbitrary pattern setup. 
This implies that the set of definable sets $(\mathbb{P}_{ext}, \subseteq)$ is not necessarily closed under intersection in an arbitrary pattern setup as shown in \figref{fig:exampleSequences}.  
One should also note that, conversely to pattern setups where some subsets can be even non-coverable (see Example~\ref{example:definableandcoverable}), In Pattern Structures, all subsets $A \subseteq \mathcal{G}$ are coverables since $cov(A) \supseteq cov(\mathcal{G}) \neq \emptyset$. In fact, $\mathcal{G}$ is always definable since $ext\left(\bigsqcap \delta[\mathcal{G}]\right) = \mathcal{G}$.
\end{note}

%% file: proofs/proof9.tex
\begin{proof}
This result is straightforward from the fact that $ext \circ int$ is a closure operator. 
Indeed, according to Theorem 1 in \cite{GanterW99} (page 8), we have $\mathbb{P}_{ext} = \{ext \circ int ( A ) \mid A \subseteq \mathcal{G}\}$ is closure system. By application of the theorem we have:
$
ext \circ int : A \mapsto \bigcap \{E \in \mathbb{P}_{ext} \mid A \subseteq E\}
$.
\end{proof}

%% file: figures/hasse_diagram_powerset_abc_without_tag.tex
\begin{tikzpicture}[scale=.7]
\small
\node (0) at (0,0) {$\emptyset$};
\node (1) at (-2,1) {$\mathbf{\{a\}}$};
\node (2) at (0,1) {$\{b\}$};
\node (3) at (2,1) {$\{c\}$};
\node (12) at (-2,2) {$\{a,b\}$};
\node (13) at (0,2) {$\{a,c\}$};
\node (23) at (2,2) {$\mathbf{\{b,c\}}$};
\node (123) at (0,3) {$\mathbf{\{a,b,c\}}$};
\draw (0) -- (1) -- (12) -- (123) -- (23) -- (3) -- (0);
\draw (0) -- (2) -- (23);
\draw (2) -- (12);
\draw (3) -- (13) -- (1);
\draw (13) -- (123);
\end{tikzpicture}

%% file: figures/hassediagramitemsets.tex
\begin{tikzpicture}[scale=.7]
\node (1) at (0,-1) {$(\{g_1\}, \{a,b,c\})$};
\node (123) at (-2.4,0.5) {$(\{g_1, g_2, g_3\}, \{a\})$};
\node (14) at (2.4,0.5) {$(\{g_1, g_4\}, \{b, c\})$};
\node (1234) at (0,2) {$(\{g_1,g_2,g_3,g_4\}, \emptyset)$};
\draw (1) -- (123) -- (1234) -- (14) -- (1);
\end{tikzpicture}

%% file: proofs/proof10.tex
\begin{proof}
Let us show both implications for a finite $\mathcal{G}$:
\begin{itemize}
\item[$\mathbf{\Rightarrow}$] The empty set has a meet in $\underline{\mathcal{D}}$ since $\delta[\emptyset] = \emptyset$ has a meet. Thus $\underline{\mathcal{D}}$ has a top element $\top = \bigsqcup \mathcal{D} = \bigsqcap \emptyset$. Moreover, let $S \subseteq \mathcal{D}$ be a finite set, one can build a finite set $\mathcal{G}$ such that $\delta[\mathcal{G}] = S$. Since $\mathbb{P}$ is a pattern structure then $S = \delta[\mathcal{G}]$ has a meet. We conclude that $\underline{\mathcal{D}}$ is an upper-bounded meet-semilattice.
\item[$\mathbf{\Leftarrow}$] Let $\mathbb{P} = (\mathcal{G}, \underline{\mathcal{D}}, \delta)$ be a pattern setup. Any subset of $\delta[\mathcal{G}]$ is finite subset of $\mathcal{D}$ and thus has a meet (including the $\emptyset$ since $\mathcal{D}$ has its top element).
\end{itemize}

Let us now consider the case of arbitrary set $\mathcal{G}$:
\begin{itemize}
\item[$\mathbf{\Rightarrow}$] Let $S \subseteq \mathcal{D}$, one can build $\mathcal{G}$ such that $\delta[\mathcal{G}] = S$. Since $\mathbb{P}$ is a pattern structure then $\delta[\mathcal{G}] = S$ has a meet. We conclude that $\underline{\mathcal{D}}$ is a complete lattice.
\item[$\mathbf{\Leftarrow}$] Let $\mathbb{P} = (\mathcal{G}, \underline{\mathcal{D}}, \delta)$ be a pattern setup. Any subset of $\delta[\mathcal{G}]$ is a subset of $\mathcal{D}$ and thus has a meet (including the $\emptyset$ since $\mathcal{D}$ has its top element).
\end{itemize}
This concludes the proof.
\end{proof}

%% file: sections/fromClosedToSupportClosedDescriptions.tex
\newpage
\section{From Closed Patterns to Support-Closed Patterns\label{sec:supportclosed}}
In pattern mining, another notion of closedness is generally considered \citep{DBLP:conf/sdm/YanHA03, DBLP:conf/icde/WangH04}. Definition \ref{def:supportClosedDescriptions} defines formally such a notion dubbed \emph{support-closedness} by \cite{DBLP:journals/tcs/BoleyHPW10}. 

\begin{definition}
\label{def:supportClosedDescriptions}
A description $d$ is said to be \empha{support-closed} in a pattern setup iff:
\begin{eqnarray*}
(\forall c \in \mathcal{D}) \:\: d \sqsubsetneq c \Longrightarrow ext(c) \subsetneq ext(d).
\end{eqnarray*}
\end{definition}

We will see below that this notion is linked to maximal common descriptions.

\subsection{On Maximal Common Descriptions}
In  pattern structures, support-closed patterns coincide exactly with closed descriptions (i.e. fixpoints of $int \circ ext$) since $int$ takes a subset of objects to the greatest common description. However, when we consider an arbitrary pattern setup, such a maximum common description may not exist (see Example~\ref{example:patternstructure:presentation2}). One straightforward generalization is to associate to a subset of object the set of its maxim\underline{al} common descriptions (see Definition \ref{def:maximalcommondescriptions}). Proposition~\ref{prop:support-closed-and-maximal-common} builds then a bridge between support-closed patterns and maximal common descritpions.    


\begin{definition}
\label{def:maximalcommondescriptions}
The set of \empha{maximal covering (common) descriptions} of a subset $A \subseteq \mathcal{G}$, denoted by $cov^*(A)$, is given by:
\begin{eqnarray*}
cov^*: A \mapsto  max(cov(A)) = max(\delta[A]^\ell)
\end{eqnarray*}
\end{definition}

\begin{proposition}\label{prop:support-closed-and-maximal-common}
A description $d \in \mathcal{D}$ is \emph{support-closed} \emph{iff} $(\exists A \subseteq \mathcal{G})$ $d \in cov^*(A)$. 
The set of all \emph{support-closed descriptions} is given by:
$
\mathcal{D}^* = \bigcup_{A \subseteq \mathcal{G}} cov^*(A)
$.
\end{proposition}

\begin{mdframed}
\input{proofs/proof12.tex}
\end{mdframed}

\begin{example}
Reconsider Example~\ref{example:patternsetup:presentation}, we have $cov(\{g_2,g_4\}) = \{``b", ``bb", "c"\}$. Hence, the maximal covering ones are given by $cov^*(\{g_2,g_4\}) = \{``bb", "c"\}$. 
\end{example}

\begin{note}
In a pattern structure we have:
$
\lfunction{cov^*}{\wp(\mathcal{G})}{\wp(\mathcal{D})}{A}{\{int(A)\}}
$.
\end{note}

\subsection{On Upper-Approximation Extents}
Going back to pattern structures, the closure operator $ext \circ int$ takes any subset $A \subseteq \mathcal{G}$ to the smallest definable set $ext(int(A))$ enclosing it as stated by Proposition~\ref{prop:extintclosure}. This fact is used to enumerate all definable sets via the closure operator (see \cite{Kuznetsov1993, DBLP:conf/pkdd/Kuznetsov99}). 
From Rough Set Theory \citep{DBLP:journals/ijpp/Pawlak82} perspective, the set $ext(int(A))$ can be seen as the upper approximation of an arbitrary and potentially non definable set $A$ in $\mathbb{P}_{ext}$. However, when it comes to an arbitrary pattern setup, a non-definable set $A$ may have many \emph{minimal} definable sets enclosing it or no one if it is non-coverable (see Example \ref{rexample:upperapproxandcovI}). Definition \ref{def:upperapproximations} formalizes this second generalization.


\begin{definition}
\label{def:upperapproximations}
The \empha{set of upper-approximation extents} of a subset $A \subseteq \mathcal{G}$, denoted by $\overline{A}$, is given by the set of minimal definable sets in $\mathbb{P}_{ext}$ enclosing $A$:
\begin{eqnarray*}
\overline{A} = min(\{E \in \mathbb{P}_{ext} \mid A \subseteq E \}) = min(\uparrow A \cap \mathbb{P}_{ext}).
\end{eqnarray*}
\end{definition}

\begin{example}\label{rexample:upperapproxandcovI}
Consider Example~\ref{example:patternsetup:presentation}, the upper approximations of subset $A = \{g_2,g_4\}$ is given by $\overline{A} = \{A\}$ since $A$ is definable. For the set $B = \{g_1,g_2\}$, we have  $\overline{B} = \{\{g_1, g_2, g_3\}, \{g_1, g_2, g_4\}\}$ that is $B$ has two upper-approximation extents. For $C = \{g_3, g_4\}$, it is clear that there is no definable set in $\mathbb{P}_{ext}$ enclosing $C$ (see \figref{fig:exampleSequences} (right)), thus $\overline{C} = \emptyset$. 
\end{example}

\begin{note}
According to Proposition \ref{prop:compositionbetweenextandcov}, we have $(\forall A \subseteq \mathcal{G})$
$
\overline{A} = min(ext[cov(A)])
$. 
Moreover, in a pattern structure, $\overline{A} = \{ext(int(A))\}$ for all $A \subseteq \mathcal{G}$.
\end{note}

\subsection{Linking Upper-Approximation Extents to Support-Closed Patterns.}
We have seen before that on the one hand $cov^*$ operator is somehow a generalization of pattern structure \emph{$int$ operator} in an arbitrary pattern setup and on the other hand, \emph{upper-approximation extents} operator is a generalization of pattern structure closure operator $ext \circ int$. Indeed, in a pattern structure we have for $A \subseteq \mathcal{G}$, $cov^*(A) = \{int(A)\}$ and $\overline{A} = \{ext(int(A))\}$. That is:
$\overline{A} = ext[cov^*(A)]$.
One judicious question is that, does this property still hold for an arbitrary pattern setup?
Let us analyze the following example.

\begin{example}\label{example:upperapproxandmaximalcover}
Reconsider Example \ref{example:patternsetup:presentation} and the definable set $A = \{g_2,g_4\}$. We have $cov(A) = \{``b", ``bb", ``c"\}$ and $cov^*(A) = \{``bb", ``c"\}$. Thus on the one hand, we have:
\begin{eqnarray*}
ext[cov^*(A)] & = & \{ext(``bb"), ext(``c")\} = \{\{g_2,g_4\},\{g_1,g_2,g_4\}\}
\end{eqnarray*}
but on the other hand, since $A$ is definable, we have $\overline{A} = \{\{g_2,g_4\}\}$. 
\end{example}

Hence, according to Example \ref{example:upperapproxandmaximalcover}, it is clear that the property $\overline{A} = ext[cov^*(A)]$ does not hold in a pattern setup. In fact, things can go even worse when the description space is an \emph{arbitrary infinite poset}. Let us analyze for that a second Example: 

\begin{example}
\label{example:holesindescriptions}
Let be the poset $(\mathcal{D}, \sqsubseteq)$ presented in \figref{fig:notmeetsemimultilattice} where:
\begin{itemize}
\item $\mathcal{D} = \{a,b\} \cup \{c_i \mid i\in \mathbb{N}\}$, 
\item $(\forall i \in \mathbb{N})$ $c_i \sqsubseteq c_{i+1}$,  $c_i \sqsubseteq a$ and $c_i \sqsubseteq b$.
\end{itemize}

Let be the pattern setup $\mathbb{P} = (\mathcal{G}, (\mathcal{D}, \sqsubseteq), \delta)$ such that $\mathcal{G} =\{g_1, g_2\}$, $\delta(g_1) = a$ and $\delta(g_2) = b$. It is clear that  
$
cov(\{g_1, g_2\}) = \{c_i \mid i \in \mathbb{N}\}
$.
Yet, $cov^*(\{g_1, g_2\}) = \emptyset$ since $cov(\{g_1, g_2\})$ is an infinitely ascending chain. Therefore, there is no maximal common description covering both $g_1$ and $g_2$.  
Thus, given $A \subseteq \mathcal{G}$ it seems that there is no link between $cov^*(A)$ and $\overline{A}$ in an arbitrary pattern setup $(\mathcal{G}, \underline{\mathcal{D}}, \delta)$. In fact, this even mean that considering only maximal common descriptions to look for all possible definable sets is totally a wrong idea since maximal covering descriptions do not hold all the information about definable sets. Indeed, while $ext[\mathcal{D}] = \{\{g_1\},\{g_2\},\{g_1,g_2\}\}$, the extents obtained from the set of support-closed patterns $\mathcal{D}^*$ is given by:
\begin{eqnarray*}
ext[\mathcal{D}^*] = \{\{g_1\},\{g_2\}\} \subsetneq ext[\mathcal{D}]
\end{eqnarray*}
\end{example}

Going back to the case of pattern structures, there is a strong link between $int$ and $cov$. Indeed, $(\forall A \subseteq \mathcal{G})$ $cov(A) = \downarrow int(A)$. In other words, knowing the \emph{intent} of a subset of objects allows us to know every single common description to all objects in $A$. If we had want to generalize such a property for an arbitrary pattern setup we would have expected:
$(\forall A \subseteq \mathcal{G})$ $cov(A) = \downarrow cov^*(A)$.
In the sense that knowing maximal covering descriptions allows to deduce all the covering ones. Such a property does not hold for any pattern setup as shown in Example~\ref{example:holesindescriptions}. This property is directly linked to what is called \emph{multilattices} which we revisit in the next section. 

%% file: proofs/proof12.tex
\begin{proof}
We prove the two implications: 
\begin{itemize}
\item[$(\pmb{\Rightarrow})$] Let $d \in \mathcal{D}$ be a support-closed description and let $A = ext(d)$. Hence, according to proposition \ref{prop:compositionbetweenextandcov} we have $d \in cov(A)$. Let us show now that $d \in cov^*(A)$. Suppose that $d \not\in cov^*(A)$ that is $\uparrow d \cap cov(A) \neq \{d\}$. Since $d \in cov(A)$, there is then at least $c \in cov(A)$ such that $d \sqsubsetneq c$. Thus, in one hand and according to proposition \ref{prop:orderreversingmappints}, $ext(c) \subseteq ext(d)$. And since $c \in cov(A)$, according to proposition \ref{prop:compositionbetweenextandcov}, $ext(c) \supseteq ext(d)$. Thus $ext(c) = ext(d)$. This is contradictory with the fact that $d$ is support-closed ($\exists c \in \mathcal{D}$ s.t. $ d \sqsubsetneq c$ and $ext(c) = ext(d)$). 
Therefore, $(\exists A \subseteq \mathcal{G})$  $d \in cov^*(A)$.

\item[$(\pmb{\Leftarrow})$]
Suppose that $\exists A \subseteq \mathcal{G}$ s.t. $d \in max(cov(A))$. According to proposition \ref{prop:compositionbetweenextandcov} and since $d \in cov(A)$, we have $A \subseteq ext(d)$. 
Let now be $c \in \mathcal{D}$ such that $d \sqsubsetneq c$, we have $c \not\in cov(A)$ since $d$ is maximal in $cov(A)$. According to proposition \ref{prop:orderreversingmappints} we have $ext(c) \subseteq ext(d)$. 
Moreover, using proposition~\ref{prop:compositionbetweenextandcov} and since $c \not\in cov(A)$ we have $A \not\subseteq ext(c)$. Since $A \subseteq ext(d)$ then $ext(c) \neq ext(d)$. Thus $\forall c \in \mathcal{D}$ such that $d \sqsubsetneq c$ we have $ext(c) \subsetneq ext(d)$; that is $d$ is support-closed.
%
%
%
%
%
\end{itemize}
The formula of $\mathcal{D}^*$ is deduced directly. Please notice also that if there exists $A$ s.t. $d \in cov^*(A)$ then $d \in cov^*(ext(d))$ (use $(\pmb{\Leftarrow})$ then $(\pmb{\Rightarrow})$). Hence, $d \in \mathcal{D}$ is support-closed iff $d \in cov^*(ext(d))$. 
\end{proof}

%% file: sections/multilattices.tex
\section{Multilattices\label{sec:multilattices}}

The term \empha{multilattice} was introduced for the first time by \cite{benado1955ensembles}. This notion have not received much interest for a long period, but have been unearthed and revisited in the beginning of the $21^\text{st}$ century by \cite{martinez2001multilattices,cordero2004new,DBLP:journals/dm/MartinezGGC05} for other purposes. We will start here by presenting multilattices following Martinez's et al. We will then understand the main difference between Benado's multilattices and (Martinez's et al.) multilattices afterward.

%
Before giving the formal definition of multilattices, we start by defining the notion of multi-infimum and multi-supremum.
In the following section $(P, \leq)$ denotes an arbitrary poset and $S \subseteq P$ denotes an arbitrary subset. 

\begin{definition}
A \empha{multi-infimum} (resp. \empha{multi-supremum}) of $S$ is a maximal (resp. minimal) element of $S^\ell$ (resp. $S^u$). The set of multi-infima (resp. multi-suprema) of $S$ is denoted by \minf(S) (resp. \msup(S)) and:
\begin{eqnarray*}
\minf(S) = max(S^\ell) & & \msup(S)  = min(S^u)
\end{eqnarray*}

We say then that:  
\begin{itemize}[topsep=1pt]
    \item $S$ \empha{has all its multi-infima} iff: $S^\ell = \downarrow max(S^\ell) = \downarrow \minf(S)$.
    \item $S$ \empha{has all its multi-suprema} iff: $S^u = \uparrow min(S^u) = \uparrow \msup(S)$. 
\end{itemize}
\end{definition}

Multilattices, as their names imply, are related in their definition with lattices. Simply put, multilattices are a relaxation of lattices where rather than demanding that the set of lower (resp. upper) bounds  of each nonempty finite subset has its infimum (resp. supremum), multilattices demand that the set of lower (resp. upper) bounds of each nonempty finite subset has all its multi-infima (resp. multi-suprema). 

%
%

\begin{definition}
\label{def:multilattice}
A poset $(P, \leq)$ is said to be: 
\begin{itemize}
    
    \item A \empha{meet-multisemilattice} if for all nonempty finite subsets $S \subseteq P$, $S$ has all its mutli-infima. 
    
    \item A \empha{join-multisemilattice} if for all nonempty finite subsets $S \subseteq P$, $S$ has all its mutli-suprema.
    
    \item A \empha{multilattice} if it is both a meet-multisemilattice and a join-multisemilattice.   
    
    \item A \empha{complete meet-multisemilattice} if for all subsets $S \subseteq P$, $S$ has all its mutli-infima. 
    
    \item A \empha{complete join-multisemilattice} if for all subsets $S \subseteq P$, $S$ has all its mutli-suprema.
    
    \item A \empha{complete multilattice} if it is both a complete meet-multisemilattice and a complete join-multisemilattice.   
\end{itemize}
\end{definition}

It is clear that all finite posets, or more generally finite-chain posets, are complete multilattices. More precisely, posets having the ascending (resp. descending) chain condition are complete join-multisemilattices (resp. meet-multisemilattices). It is also clear that all (complete) (semi)lattices are (complete) multi(semi)lattices since requiring that a subset $S \subseteq P$ to have all its multi-infima is weaker than requiring it to have an infimum. One should note also that since $\emptyset^\ell = \emptyset^u = P$ then we have:
\begin{itemize}
    \item If $(P,\leq)$ is a {complete \meetmultisemilattice{}} then $P = \downarrow max(P)$.
    
    \item If $(P,\leq)$ is a {complete \joinmultisemilattice{}} then $P = \uparrow min(P)$.
\end{itemize}

\medskip

When we compare with lattices (cf. note \ref{note:latticesproperty}), two questions straightforwardly raise:
\begin{itemize}
    \item Are pair of elements the building blocks of a multilattice?
    \item Are all complete semimultilattices complete multilattices?
\end{itemize}
The followings sections answers negatively to these both questions.

\subsection{On Benado's Multilattices}
As said before, \emph{multilattices} ware introduced for the first time by \cite{benado1955ensembles}. However, Benado defined multilattices as follow: a poset $(P,\leq)$ is said to be multilattice if and only if all pairs of elements have all their multi-infima and all their multi-suprema. Formally:  
\begin{eqnarray}
\forall x,y \in P : \{x, y\}^\ell = \downarrow \minf(\{x,y\}) \text{ and } \{x, y\}^u = \uparrow \msup(\{x,y\})  
\end{eqnarray}
We will call here posets verifying such a property \empha{Benado's multilattices}. 
As clearly explained in \cite{DBLP:journals/dm/MartinezGGC05}, such a property is not sufficient to have all non empty finite subsets have their multi-infima and multi-suprema. This is in contrast to lattices where it suffices to have a meet and join for all subsets of two elements to have the meet and join for all non empty finite subsets. For instance, one could verify that the poset depicted in \figref{fig:twonotimplythreeinmultilattices} is a Benado's multilattice. Yet, it is not a multilattice according to Definition~\ref{def:multilattice}. For instance, the set $\{a,b,c\}$ does not have all its multi-infima. Indeed, $\{a,b,c\}^\ell = \{abc_i \mid i \in \mathbb{N}\}$, however, $max(\{a,b,c\}^\ell) = \emptyset$. It follows that $\{a,b,c\}^\ell \neq \downarrow max(\{a,b,c\}^\ell)$.

%

\subsection{On Complete Multisemilattices}
We have seen that all complete semilattices are complete lattices. However, this property no longer holds for complete multisemilattices.  
In fact one could have complete \meetmultisemilattice{} that is even not a \joinmultisemilattice{} and \emph{vice versa}. For instance, \figref{fig:notmeetsemimultilattice} depicts a complete \joinmultisemilattice{} that is even not a meet-semimultilattice. For instance, the set of lower bounds $\{a,b\}^\ell = \{c_i \mid i \in \mathbb{N}\}$ is an infinitely ascending chain and thus $max(\{a,b\}^\ell) = \emptyset$. Thus, $\{a,b\}^\ell \neq \downarrow max(\{a,b\}^\ell)$. In other words, $\{a,b\}$ does not have all its multi-infima.

\begin{figure}[t]
\centering
\begin{minipage}{0.16\linewidth} 
\centering
\scalebox{.72}{\input{figures/notmeetsemimultilattice.tex}}
\caption{A complete \joinmultisemilattice{} but not a \meetmultisemilattice{}\label{fig:notmeetsemimultilattice}}
\end{minipage}
\hfill
\begin{minipage}{0.8\linewidth} 
\centering
\scalebox{.7}{\input{figures/twonotimplythreemultimeet.tex}}
\caption{For all $x,y$ in this poset, $\{x,y\}$ has all its multi-infima. That is, this poset is a multistructure following \cite{benado1955ensembles}, however it is not a multilattice following our definition.\label{fig:twonotimplythreeinmultilattices}}
\end{minipage}
\end{figure}


\subsection{On chain-completeness and complete multilattices}
Complete multilattices are linked chain-complete posets. Before giving this relationship, let us recall the definition of chain-completeness.

\begin{definition}
\label{def:chain-completeness}
A poset $(P, \leq)$ is said to be: 
\begin{itemize}
\item \empha{Chain-complete} if all chains in $P$, including $\emptyset$, has its join.
\item \empha{Dually chain-complete} if all chains in $P$, including $\emptyset$, has its meet.
\item \empha{Doubly chain-complete} if it is chain-complete and dually chain-complete.
\end{itemize}
\end{definition}  

Since the empty set matches the definition of a chain-complete, all chain-complete posets are bounded. An important theorem linking complete lattices to chain-completeness is given in Theorem \ref{thm:latticesandchaincompleteness}.

\begin{theorem}[\textbf{Theorem 3.24 from \cite{roman2008lattices}}\label{thm:latticesandchaincompleteness}]
A lattice $(P, \leq)$ is a complete lattice if and only if it is chain-complete.
\end{theorem}

One straightforward question is what is the relationship between complete multilattices and chain-complete posets. The answer is given in Theorem \ref{thm:chaincompletetoMultilattices}.


\begin{theorem}\label{thm:chaincompletetoMultilattices}
Under \emph{Axiom of Choice \textbf{(AC)}} assumption\footnote{I am grateful to \textsc{Jozef P\'{o}cs} for attracting my attention to Zorn's Lemma.}, we have: 
\begin{itemize}
    \item All chain-complete posets are complete \meetmultisemilattice{}.
    \item All dually chain-complete posets are complete \joinmultisemilattice{}.
    \item All doubly chain-complete posets are complete multilattices.
\end{itemize}
\end{theorem}

\begin{mdframed}
\input{proofs/proof13.tex}

\end{mdframed}

\begin{note}
It is important to note that double chain-completeness is only a sufficient condition (under the \emph{Axiom of Choice}) to have a complete multilattice but not a necessary one. Indeed, one can show that the poset depicted in Figure \ref{fig:notchaincompleteposetbutcompletemultilattice} is a complete multilattice (Remark that $\forall i \in \mathbb{N} : c_i \leq a_i$, $c_i \leq c_{i+1}$, $a_i \leq e_0$ and $a_i \leq e_1$) but not chain-complete since the chain $C = \{c_i \mid i \in \mathbb{N}\}$ has not a join. Indeed, $C^u=\{e_0, e_1\}$ which is an antichain (i.e. $C^u$ has two minimal elements).
\end{note}

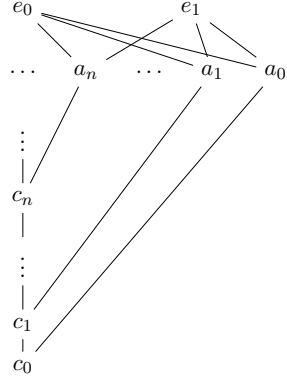
\begin{figure}[t]
\centering
\scalebox{.8}{\input{figures/notchaincompleteposetbutcompletemultilattice.tex}}
 \caption{A complete multilattice that is not chain-complete  \label{fig:notchaincompleteposetbutcompletemultilattice}}
\end{figure}

%% file: figures/notmeetsemimultilattice.tex
\begin{tikzpicture}[scale=.7]
\node (a) at (-0.5,7) {$a$};
\node (b) at (0.5,7) {$b$};
\node (inf) at (0,5.5) {$\vdots$};
\node (cn) at (0,4) {$c_n$};
\node (ci) at (0,2.5) {$\vdots$};
\node (c1) at (0,1) {$c_1$};
\node (c0) at (0,0) {$c_0$};
\draw (a) -- (inf) -- (cn) -- (ci) -- (c1) -- (c0);
\draw (b) -- (inf);
\end{tikzpicture}

%% file: figures/twonotimplythreemultimeet.tex
\begin{tikzpicture}[scale=.7]
\node (abc0) at (0,0) {$abc_0$};
\node (abc1) at (0,1) {$abc_1$};
\node (abci) at (0,2.5) {$\vdots$};
\node (abcn) at (0,4) {$abc_n$};
\node (abcinf) at (0,5.5) {$\vdots$};

\node (a) at (-4,8.5) {$a$};
\node (b) at (0,8.5) {$b$};
\node (c) at (4,8.5) {$c$};

\node (ab0) at (-11,7) {$ab_0$};
\node (ab1) at (-9.5,7) {$ab_1$};
\node (abi) at (-8,7) {$\hdots$};
\node (abn) at (-6.5,7) {$ab_n$};
\node (abinf) at (-5,7) {$\hdots$};

\node (bc0) at (11,7) {$bc_0$};
\node (bc1) at (9.5,7) {$bc_1$};
\node (bci) at (8,7) {$\hdots$};
\node (bcn) at (6.5,7) {$bc_n$};
\node (bcinf) at (5,7) {$\hdots$};

\node (ac0) at (-3,7) {$ac_0$};
\node (ac1) at (-1.5,7) {$ac_1$};
\node (aci) at (0,7) {$\hdots$};
\node (acn) at (1.5,7) {$ac_n$};
\node (acinf) at (3,7) {$\hdots$};

\draw (abc0) -- (abc1) -- (abci) -- (abcn) -- (abcinf);

\draw (abc0) -- (ab0) -- (a);
\draw (abc1) -- (ab1) -- (a);
\draw (abcn) -- (abn) -- (a);
\draw (ab0) -- (b);
\draw (ab1) -- (b);
\draw (abn) -- (b);

\draw (abc0) -- (bc0) -- (c);
\draw (abc1) -- (bc1) -- (c);
\draw (abcn) -- (bcn) -- (c);
\draw (bc0) -- (b);
\draw (bc1) -- (b);
\draw (bcn) -- (b);

\draw (abc0) -- (ac0) -- (a);
\draw (abc1) -- (ac1) -- (a);
\draw (abcn) -- (acn) -- (a);
\draw (ac0) -- (c);
\draw (ac1) -- (c);
\draw (acn) -- (c);
\end{tikzpicture}

%% file: proofs/proof13.tex
\begin{proof}
Before proving the theorem, we attract the reader to \empha{Zorn's Lemma}. This lemma states that 
if every chain in a poset $(P, \leq)$ has an upper-bound, then $(P, \leq)$ has a maximal element. Formally:
\begin{align*}
    (\forall C \in \mathscr{C}(P)) \:\:\: C^u \neq \emptyset \:\: \Longrightarrow \:\: max(P) \neq \emptyset  
\end{align*}
%
A stronger statement, yet equivalent, of Zorn's Lemma say even that: 
%
\begin{align}
    \label{eq:zorn2}
    (\forall C \in \mathscr{C}(P)) \:\:\: C^u \neq \emptyset \:\: \Longrightarrow \:\: P = \downarrow max(P)
\end{align}
%
Zorn's Lemma need to be considered as an \empha{axiom} since it is equivalent to the well-known \empha{axiom of choice (AC)}.

Let be now a chain-complete poset $(P, \leq)$, we need to show that $(P,\leq)$ is a complete \meetmultisemilattice{}. 
Let $S \subseteq P$ be an arbitrary subset of $P$. We show here that $S^\ell = \downarrow max(S^\ell)$. It is straightforward by definition and independently from any assumption that $\downarrow max(S^\ell) \subseteq S^\ell$. It remains to show that $S^\ell \subseteq \downarrow max(S^\ell)$. 
Since $(P, \leq)$ is \emph{chain-complete}, then every $C \subseteq S^\ell$ has its join $\bigvee C \in P$. Hence, according to Lemma \ref{lemma:lowerboundsarestables} and since $C \subseteq S^\ell$ then $\bigvee C \in S^\ell$.
Thus every chain $C$ in the sub-poset $(S^\ell, \leq)$ has an upper bound $\bigvee C \in S^\ell$. According to \emph{Zorn's Lemma} (cf. equation~\ref{eq:zorn2}), we have $S^\ell = \downarrow max(S^\ell)$. Hence, $(P, \leq)$ is a complete \meetmultisemilattice{}.  

In order to demonstrate the other statements of the theorem, one can follow the same steps to show that $S^u = \uparrow min(S^u)$ using Zorn's Lemma on the dual poset of dually chain-complete posets. 
\end{proof}

%% file: figures/notchaincompleteposetbutcompletemultilattice.tex
\begin{tikzpicture}[scale=.7]
\node (c0) at (0,0) {$c_0$};
\node (c1) at (0,1) {$c_1$};
\node (ci) at (0,2.5) {$\vdots$};
\node (cn) at (0,4) {$c_n$};
\node (cinf) at (0,5.5) {$\vdots$};

\node (e0) at (0,8.5) {$e_0$};
\node (e1) at (4,8.5) {$e_1$};

\node (a0) at (6,7) {$a_0$};
\node (a1) at (4.5,7) {$a_1$};
\node (ai) at (3,7) {$\hdots$};
\node (an) at (1.5,7) {$a_n$};
\node (ainf) at (0,7) {$\hdots$};

\draw (c0) -- (c1) -- (ci) -- (cn) -- (cinf);

\draw (c0) -- (a0) -- (e0);
\draw (c1) -- (a1) -- (e0);
\draw (cn) -- (an) -- (e0);
\draw (a0) -- (e1);
\draw (a1) -- (e1);
\draw (an) -- (e1);

\end{tikzpicture}

%% file: sections/patternMultiStructures.tex
\newpage
\section{Pattern \Multistructure\label{sec:pattermultistructure}}

Definition \ref{def:patternmultistructure} proposes a new structure that lies between \emph{pattern setups} which rely only on arbitrary posets with no additional property and \emph{pattern structures} which demand a \emph{meet} for every subset of $\delta[\mathcal{G}]$.

\begin{definition}
\label{def:patternmultistructure}
A pattern setup $\mathbb{P}=(\mathcal{G}, (\mathcal{D}, \sqsubseteq), \delta)$ is said to be a \empha{pattern \multistructure} if every subset of $\delta[\mathcal{G}]$ has all its multi-infima, i.e. :
\begin{eqnarray*}
(\forall A \subseteq \mathcal{G}) \:\: cov(A) = \downarrow cov^*(A)
\end{eqnarray*}
\end{definition}

A pattern \multistructure{} adds an additional condition on a pattern setup which is the following: knowing maximal common descriptions covering all elements of a set of objects $A$ allows us to deduce using the order $\sqsubseteq$ every single covering description. Please note that using the notation of multi-infima, $cov^*(A)$ is given by $\minf(\delta[A])$.  It is clear that all pattern structures are by definition pattern multi-structures.  

Graphs ordered by subgraph isomorphism relation introduced in \cite{DBLP:conf/pkdd/Kuznetsov99} induce a pattern \multistructure{} on the set of graphs, but not a pattern structure (a pattern structure is induced on sets of graphs incomparable wrt. subgraph isomorphism). Same remark holds for sequential patterns \citep{DBLP:journals/ijgs/BuzmakovEJKNR16,DBLP:conf/icfca/CodocedoBKBN17}. This is under the assumption of the existence of a largest element $\top$ subsumed by all sequences/graphs (see Example \ref{example:patternmultistructuresequence}). 

\begin{note}
Note that in a pattern \multistructure{} the empty set $\emptyset \subseteq \mathcal{D}$ has all its \emph{multi-infima}. Since $\emptyset^\ell = \mathcal{D}$, the set $\mathcal{D}$ has all its maximal elements (i.e. $\mathcal{D} = \downarrow max(\mathcal{D})$) or in other words every chain in $(\mathcal{D}, \sqsubseteq)$ is upper-bounded.
\end{note}

\begin{example}\label{example:patternmultistructuresequence}
Reconsider the pattern setup presented in Example~\ref{example:patternsetup:presentation}. Since only finite sequences are considered in the description space, we have: $cov^*(\emptyset) = \emptyset$ even if $cov(\emptyset) = \mathcal{D}$. Thus, the considered pattern setup in Example~\ref{example:patternsetup:presentation} is \emph{not a pattern \multistructure} due to the empty set (recall that for a nonempty $A \subseteq \mathcal{G}$, $\delta[A]^\ell$ is finite and thus $\delta[A]$ has all its multi-infima). The common trick to handle the empty set is to enrich $\mathcal{D}$ with an additional largest element $\top \triangleq \bigvee \mathcal{D}$ if it does not exist. In such a case, we have $cov^*(\emptyset) = \{\top\}$.
\end{example}

Let us now reconsider the question investigated at the end of section \ref{sec:supportclosed}: \emph{``What is the link between maximal covering descriptions and upper-approximations extents in a pattern \multistructure{}"}. Before stating Theorem~\ref{thm:computingUpperApproximations} answering this question, we  shall state the following Lemma.

\begin{lemma}\label{lemma:propertiesarrowsandembeddings}
Let $(P, \leq)$ and $(Q ,\leq)$ be two posets and let $f : P \to Q$ be an order-reversing mapping. We have for any $S \subseteq P$ that
$\uparrow f [\downarrow S] = \uparrow f[S]$.
\end{lemma}

\begin{mdframed}
\input{proofs/proof21.tex}
\end{mdframed}

\begin{theorem}\label{thm:computingUpperApproximations} For any \emph{pattern \multistructure} $\mathbb{P}$ we have:
\begin{eqnarray*}
(\forall A \subseteq \mathcal{G}) & & \overline{A} = min(ext[cov^*(A)]) 
\end{eqnarray*}
\end{theorem}

\begin{mdframed}
\input{proofs/proof14.tex}

\end{mdframed}

\bigskip

Another important remark related to Example~\ref{example:holesindescriptions} is the fact that the support-closed patterns in a pattern setup does not hold all the information about the definable sets. Theorem~\ref{thm:allinformationisinsupportclosed} states that this is no longer the case for pattern multistructures.

\begin{theorem}
\label{thm:allinformationisinsupportclosed}
Given a pattern multistructure $\mathbb{P}$ for which the set of support-closed patterns is $\mathcal{D}^*$ (cf.  Proposition~\ref{prop:support-closed-and-maximal-common}), we have:
\begin{eqnarray*}
\mathbb{P}_{ext} & = & ext[\mathcal{D}^*]
\end{eqnarray*}
\end{theorem}

\begin{mdframed}
\input{proofs/proof15.tex}
\end{mdframed}

\bigskip

Similarly to Theorem \ref{thm:upperboundedmeetsemiolattice} for pattern structures, Theorem \ref{thm:upperboundedmeetsemihyperlattice} connects  \emph{multilattices} with \emph{pattern \multistructure s}. It state that (complete) \meetmultisemilattice s are to pattern \multistructure s what (complete) lattices are to pattern structures. 


\begin{theorem}\label{thm:upperboundedmeetsemihyperlattice}
Let $\underline{\mathcal{D}} = (\mathcal{D}, \sqsubseteq)$ be a poset, the following properties are equivalent:
\begin{itemize}
\item For any finite set $\mathcal{G} \neq \emptyset$ and any $\delta \in \mathcal{D}^\mathcal{G}$, $(\mathcal{G},\underline{\mathcal{D}}, \delta)$ is a pattern \multistructure{}.
\item $\underline{\mathcal{D}}$ is a \meetmultisemilattice{} having all its maximal elements (i.e. $\mathcal{D} = \downarrow max(\mathcal{D})$)  
\end{itemize}
The following properties are equivalent:
\begin{itemize}
\item For any set $\mathcal{G} \neq \emptyset$ and any $\delta \in \mathcal{D}^\mathcal{G}$, $(\mathcal{G},\underline{\mathcal{D}}, \delta)$ is a pattern \multistructure{}.
\item $\underline{\mathcal{D}}$ is a complete \meetmultisemilattice{}.
\end{itemize}
\end{theorem}

\begin{mdframed}
\input{proofs/proof16.tex}
\end{mdframed}

\bigskip

\begin{figure}[t]
\centering
\scalebox{.9}{\input{figures/problemWithPextInPatternMultistructure.tex}}
 \caption{A complete multilattice with an infinite antichain (Proof of proposition \ref{prop:problemWithPextInPatternMultistructure}) \label{fig:problemWithPextInPatternMultistructure}}
\end{figure}
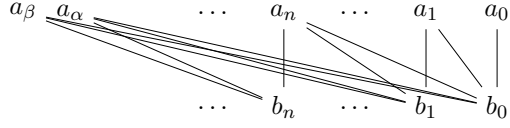

Last but not least, we have seen in section~\ref{sec:patternStructure} that in the case of a pattern structure, $(\mathbb{P}_{ext}, \subseteq)$ is a complete lattice. One can say that the property of having the infimum in the description space has been transferred to the poset of definable sets thanks to extent operator. 
When it comes to a pattern setup on finite set of objects, it is clear that $(\mathbb{P}_{ext}, \subseteq)$ is a complete multilattice since it is finite. However, does this property still hold for the case of infinite set of objects?
Unfortunately, the answer is negative as stated in Proposition \ref{prop:problemWithPextInPatternMultistructure}. This proposition tells also that not all definable sets above $A$ in a pattern multistructure are above at least one upper-approximation of $A$. 

\begin{proposition}\label{prop:problemWithPextInPatternMultistructure}
There exists a pattern \multistructure{} $\mathbb{P} = (\mathcal{G}, (\mathcal{D}, \sqsubseteq), \delta)$ such that $(\mathbb{P}_{ext}, \subseteq)$ is not a \joinmultisemilattice{} in which:
$(\exists A \subseteq \mathcal{G}) \:  \uparrow \overline{A} \cap \mathbb{P}_{ext} \neq \uparrow A \cap \mathbb{P}_{ext}$.
\end{proposition}

\begin{mdframed}
\input{proofs/proof17.tex}
\end{mdframed}

%% file: proofs/proof21.tex
\begin{proof}
Recall that $\uparrow$ and $\downarrow$ are closure operator (cf. Lemma \ref{lemma:downandupclosure}).
Let us start by showing that $\uparrow f[S] \subseteq \uparrow f[\downarrow S]$. Since $S \subseteq \downarrow S$, we have $f[S] \subseteq f[\downarrow S]$. Since $\uparrow$ is monotone, we have $\uparrow f[S] \subseteq \uparrow f[\downarrow S]$. 
It remains to show that $\uparrow f[\downarrow S] \subseteq \uparrow f[S]$. Let $u \in \uparrow f[\downarrow S]$, that is $\exists v \in f[\downarrow S]$ s.t. $v \sqsubseteq u$. Since $v \in f[\downarrow S]$, then $\exists x \in \downarrow S$ s.t. $v = f(x)$. Hence $\exists y \in S$ s.t. $x \leq y$. Using the fact that $f$ is an anti-embedding, we obtains that $f(y) \sqsubseteq f(x) \sqsubseteq u$. In other words, $\exists w \in f[S]$ s.t. $w \sqsubseteq u$. This is equivalent to say that $u \in \uparrow f[S]$. We conclude hence that $\uparrow f[\downarrow S] \subseteq \uparrow f[S]$. 
\end{proof}

%% file: proofs/proof14.tex
\begin{proof}
The proof of the theorem is a straightforward application of Lemma \ref{lemma:propertiesminandarrows} and Lemma \ref{lemma:propertiesarrowsandembeddings}.
Let $A \subseteq \mathcal{G}$, since $\mathbb{P}$ is a pattern \multistructure{}, then:
\begin{eqnarray*}
cov(A)= \downarrow max(cov(A)) & & \Rightarrow 
ext[cov(A)] = ext[\downarrow max(cov(A))] \\
& & \Rightarrow \uparrow ext[cov(A)] = \uparrow ext[\downarrow max(cov(A))] 
\end{eqnarray*}

Since $ext: \mathcal{D} \to \wp(\mathcal{G})$ is an order reversing, then using Lemma \ref{lemma:propertiesarrowsandembeddings} we have:
\begin{eqnarray*}
\uparrow ext[\downarrow max(cov(A))] =  \uparrow ext[max(cov(A))]  & & \Rightarrow \\
\uparrow ext[cov(A)]= \uparrow ext[max(cov(A)) & \Rightarrow \\
min(\uparrow ext[cov(A)]) = min(\uparrow ext[max(cov(A)]) & 
\end{eqnarray*}

Using Lemma \ref{lemma:propertiesminandarrows} we obtain
$
min(ext[cov(A)]) = min(ext[max(cov(A)))
$.
That is,
$
\overline{A} = min(ext[max(cov(A)))
$.
\end{proof}

%% file: proofs/proof15.tex
\begin{proof}
Recall that $\mathcal{D}^* = \bigcup_{B \subseteq \mathcal{G}} ext[cov^*(B)]$. Since $\mathcal{D}^* \subseteq \mathcal{D}$ and by definition $\mathbb{P}_{ext} = ext[\mathcal{D}]$. It is clear that $ext[\mathcal{D}^*] \subseteq ext[\mathcal{D}]$. It remains to show that $ext[\mathcal{D}] \subseteq ext[\mathcal{D}^*]$. Let $A \in ext[\mathcal{D}]$, since $\mathbb{P}$ is a pattern multistructure then $cov^*(A) = \downarrow cov(A)$. Let $d \in cov(A)$ s.t. $A = ext(d)$ (we have $A \in ext[\mathcal{D}]$). Since $\mathbb{P}$ is a pattern multistructure then we have a support-closed pattern $d^* \in cov^*(A) \subseteq \mathcal{D}^*$ s.t. $d \sqsubseteq d^*$. Hence, $ext(d^*) \subseteq ext(d)$. Moreover, since $cov^*(A) \subseteq cov(A)$, we have $d^* \in cov(A)$. Therefore, $A = ext(d) \subseteq ext(d^*)$. We obtain thus $A = ext(d) = ext(d^*)$, that is $A \in ext[\mathcal{D}^*]$.
\end{proof}

%% file: proofs/proof16.tex
\begin{proof}
Recall that $\mathbb{P} = (\mathcal{G}, \underline{\mathcal{D}}, \delta)$ is a pattern \multistructure{} \emph{iff} any subset $S \subseteq \delta[\mathcal{G}]$ has all its multi-infima. The proof of this theorem is almost the same as the one of Theorem \ref{thm:upperboundedmeetsemiolattice} where the existence of the meet is replaced by the existence of all multi-infima. 
\end{proof}

%% file: figures/problemWithPextInPatternMultistructure.tex
\begin{tikzpicture}[scale=.7]
\node (a0) at (6,2) {$a_0$};
\node (a1) at (4.5,2) {$a_1$};
\node (ai) at (3,2) {$\hdots$};
\node (an) at (1.5,2) {$a_n$};
\node (ainf) at (0,2) {$\hdots$};
\node (aalpha) at (-3,2) {$a_\alpha$};
\node (abeta) at (-4,2) {$a_\beta$};

\node (b0) at (6,0) {$b_0$};
\node (b1) at (4.5,0) {$b_1$};
\node (bi) at (3,0) {$\hdots$};
\node (bn) at (1.5,0) {$b_n$};
\node (binf) at (0,0) {$\hdots$};

\draw (b0) -- (a0);
\draw (b0) -- (a1);
\draw (b0) -- (an);
\draw (b0) -- (aalpha);
\draw (b0) -- (abeta);

\draw (b1) -- (a1);
\draw (b1) -- (an);
\draw (b1) -- (aalpha);
\draw (b1) -- (abeta);

\draw (bn) -- (an);
\draw (bn) -- (aalpha);
\draw (bn) -- (abeta);
\end{tikzpicture}

%% file: proofs/proof17.tex
\begin{proof}
Consider the pattern setup $\mathbb{P} = (\mathcal{G}, (\mathcal{D}, \sqsubseteq), \delta)$ where $(\mathcal{D}, \sqsubseteq)$ is the complete multilattice depicted in Fig.~\ref{fig:problemWithPextInPatternMultistructure}. We have: 
\begin{itemize}
\item $(\forall i,j \in \mathbb{N})$ $i \leq j \:\: \Leftrightarrow \:\: b_i \sqsubseteq a_j$.
\item $(\forall i \in \mathbb{N})$ $b_i \sqsubseteq a_\alpha \text{ and } b_i \sqsubseteq a_\beta$.
\end{itemize}

Since $(\mathcal{D}, \sqsubseteq)$ is a complete multilattice (i.e. it is chain-finite), then $\mathbb{P}$ is a pattern multistructure. Consider now an infinite set $\mathcal{G} = \{g_i \mid i\in\mathbb{N}\} \cup \{g_\alpha, g_\beta\}$. The mapping $\delta$ is given by: $\delta(g_\alpha)$ $=$ $a_\alpha$,  $\delta(g_\beta)$ $=$ $a_\beta$ and $(\forall i \in \mathbb{N})$ $\delta(g_i)$ $=$ $a_i$.
%
To show that the poset $(\mathbb{P}_{ext}, \subseteq)$ is not a  \joinmultisemilattice{} we need to consider two definable sets in $\mathbb{P}_{ext}$ and show that the set of their common upper-bounds in $\mathbb{P}_{ext}$ does not have all its minimal elements. Let us compute $ext$ for every $d \in \mathcal{D}$:
\begin{itemize}
    \item $ext(a_\alpha)$ $=$ $\{g_\alpha\}$ and $ext(a_\beta)$ $=$ $\{g_\beta\}$.
    \item $(\forall i \in \mathbb{N})$ $ext(a_i)$ $=$ $\{g_i\}$ and  $(\forall i \in \mathbb{N})$ $ext(b_i)$ $=$ $\{g_\alpha, g_\beta\} \cup \{g_j \mid j \geq i\}$.
\end{itemize}

Consider now the set of definable sets $\{\{g_\alpha\}, \{g_\beta\}\}$, it is clear that the set of their common upper-bounds (in $\mathbb{P}_{ext})$ is given by:
\begin{eqnarray*}
\{\{g_\alpha\}, \{g_\beta\}\}^u = \big\{\left\{g_\alpha, g_\beta\right\} \cup \left\{g_j \mid j \geq i\right\} \mid i \in \mathbb{N}\big\}
\end{eqnarray*}

The set of upper bounds is hence an infinitely descending chain and hence does not have a minimal element, in other words:
$
min(\{\{g_\alpha\}, \{g_\beta\}\}^u) = \emptyset
$.
Hence, $(\mathbb{P}_{ext}, \subseteq)$ is not a \joinmultisemilattice{}.
The proof of the second part of the proposition is straightforward. Indeed, consider the non-definable set $A = \{g_\alpha, g_\beta\}$. We do have:
$\uparrow A \cap \mathbb{P}_{ext} = ext[cov(A)] =  \big\{\left\{g_\alpha, g_\beta\right\} \cup \left\{g_j \mid j \geq i\right\} \mid i \in \mathbb{N}\big\}$.
Hence, $\overline{A} = min(\uparrow A \cap \mathbb{P}_{ext}) = \emptyset$. That is $\uparrow A \cap \mathbb{P}_{ext} \neq \uparrow \overline{A} \cap \mathbb{P}_{ext}$. 
\end{proof}

%% file: sections/completion.tex
\newpage
\section{Pattern Setup and Pattern Multistructure Completions\label{sec:completions}}
Example \ref{example:patternsetup:presentation} presents a pattern setup which is not a pattern structure. However, in FCA and pattern structure literature, it is recurrent to talk about sequential pattern structures \citep{DBLP:journals/ijgs/BuzmakovEJKNR16, DBLP:conf/icfca/CodocedoBKBN17}. In fact, instead of sequences, sets of sequences are considered which induce a richer description space. Same trick has been even used in the first paper introducing pattern structures \citep{DBLP:conf/iccs/GanterK01} concerning graph description space ordered by subgraph isomorphism. Such a technique that embeds a poset into another is called a \emph{completion} (see Definition~\ref{def:morhismproperties}). 
Different natural completions exist in the literature. For instance, the Dedekind-MacNeille completion \citep{davey2002introduction} takes an arbitrary poset to the smallest complete lattice containing it.  
The usual trick used in FCA and Pattern Structure literature to augment a base pattern setup to a pattern structure is tightly linked to the antichain completion presented below.


\begin{definition}
The \empha{antichain completion} of $(P, \leq)$ is the poset $(\mathscr{A}(P), \leqq)$ s.t.:
\begin{itemize}
\item $\mathscr{A}(P)$ is the set of all antichains of $(P, \leq)$.
\item The order $\leqq$ is given by $(\forall A, B \in \mathscr{A}(P))$ $A \leqq B \Leftrightarrow \downarrow A \subseteq \downarrow B$.\footnote{Note that $\leqq$ does not induce an order in $\wp(P)$, but just a \emph{pre-order}, since the \emph{anti-symmetry} does not hold  (see~\cite{crampton2001completion}). Indeed, consider poset $(\{a,b\}, \leq)$ where $a \leq b$. Since $\downarrow \{a,b\} = \downarrow \{b\} = \{a,b\}$, we have $\{a,b\} \leqq \{b\}$ and $\{b\} \leqq \{a,b\}$. Yet, $\{a,b\} \neq \{b\}$. Therefor, $\leqq$ does not induce an antisymmetric relation on $(\wp(\{a,b\}), \leqq)$ but it is still reflexive and transitive (i.e. $\leqq$ induce a preorder on $\wp(\{a,b\})$. }
\item The order embedding $\varphi$ from $(P, \leq)$ to $(\mathscr{A}(P), \leqq)$ is given by
\begin{eqnarray*}
\varphi : P \to \mathscr{A}(P), a \mapsto \{a\} 
\end{eqnarray*}
\end{itemize}
\end{definition}

\cite{BoldiVigna2016} and \cite{crampton2001completion} had discussed the properties of such a completion. In fact, when $P$ have the \textbf{ACC}, $(\mathscr{A}(P), \leq)$ is a \emph{distributive lattice}, where the meet and the join are given by $S_1 \wedge S_2 = max(\downarrow S_1 \cap \downarrow S_2)$ and $S_1 \vee S_2 = max(S_1 \cup S_2)$, respectively. Moreover, $(\mathscr{A}(P), \leqq)$ is always a $\vee$-semilattice whatever the nature of the poset $(P, \leq)$,  but not necessarily a lattice. \cite{BoldiVigna2016} formulated a sufficient and necessary condition in order to have $(\mathscr{A}(P), \leqq)$ be a lattice:
$\forall A,B \in \mathscr{A}(P) \:\: \exists C \in \mathscr{A}(P) \:\: \downarrow A \cap \downarrow B = \downarrow C$.


We take the opportunity here to underline an important link between the antichain completion and multilattices. Before expliciting this link in Propsosition~\ref{prop:antichaintomultilattices}, let us take a close look to the following Lemma.

\begin{lemma}\label{lemma:antichainsandmaxelements}
Let $(P, \leq)$ be an arbitrary poset and let $\mathscr{A}(P)$ be the set of its antichain. We have $\forall S \subseteq P$ : 
$
(\exists C \in \mathscr{A}(P))\:\:S = \downarrow C \Rightarrow C = max(S)
$.
\end{lemma}

\begin{mdframed}
\input{proofs/proof22.tex}
\end{mdframed}

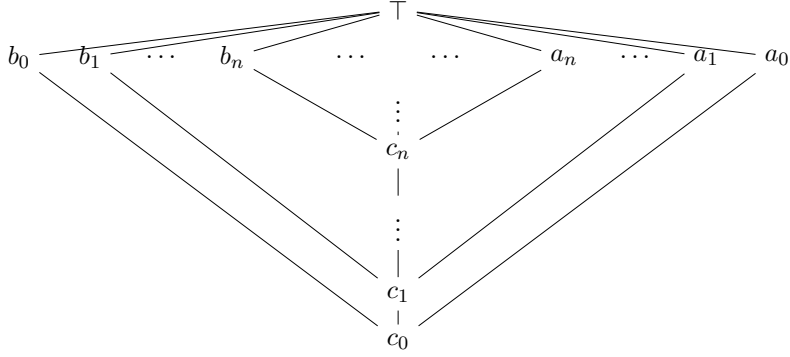
\begin{figure}[t]
\centering
\scalebox{.9}{\input{figures/completelatticenotgoodantichaincompletion.tex}}
\caption{The Antichain completion of this complete lattice is not even a meet-semilattice  \label{fig:completelatticenotgoodantichaincompletion}}
\end{figure}

\begin{proposition}\label{prop:antichaintomultilattices}
Let $(P, \leq)$ be an arbitrary poset and let $(\mathscr{A}(P), \leqq)$ be its antichain completion: If $(\mathscr{A}(P), \leqq)$ is a lattice then $(P, \leq)$ is a \meetmultisemilattice{}. Moreover, if $(\mathscr{A}(P), \leqq)$ has a top element then $P = \downarrow max(P)$. 
\end{proposition}

\begin{mdframed}
\input{proofs/proof18.tex}
\end{mdframed}

\begin{note}
One should note that the converse of Proposition~\ref{prop:antichaintomultilattices}. In fact, one can create complete lattices for which the antichain completion is not even a lattice. 
\figref{fig:completelatticenotgoodantichaincompletion} depicts such a complete lattice. Indeed, for antichains $A = \{a_i \mid i \in \mathbb{N}\}$ and $B = \{b_i \mid i \in \mathbb{N}\}$, we have $\downarrow A \bigcap \downarrow B = \{c_i \mid i \in \mathbb{N}\}$.
Hence, there is no antichain $D \in \mathscr{A}(P)$ s.t. $\{c_i \mid i \in \mathbb{N}\} =  \downarrow D$ making the antichain completion not a meet-semilattice.
\end{note}

\subsection{On Pattern Setups Antichain Completions}
The main purpose of transforming a pattern setup to another one is to augment it to a pattern structure in order to use the different results related to this latter structure. We define below the most common trick used in the FCA literature which can be called \emph{pattern setup antichain completion}.  

\begin{definition}\label{def:antichaincompletion}
Let $\mathbb{P}=(\mathcal{G}, \underline{\mathcal{D}}, \delta)$ be a \emph{pattern setup}, the \empha{antichain completion} of $\mathbb{P}$ is the \underline{\emph{pattern setup}} denoted by $\mathbb{P}^{\triangledown}$ and given by:
\begin{eqnarray*}
\mathbb{P}^{\triangledown} = \left(\mathcal{G}, (\mathscr{A}(\mathcal{D}),\leqq), \delta^{\triangledown} :g \mapsto \,\{\delta(g)\}\right)
\end{eqnarray*}
\end{definition}

Earlier in this section, we have mentioned that $(\mathscr{A}(\mathcal{D}),\leqq)$ is a lattice when $(\mathcal{D}, \sqsubseteq)$ is a finite poset (i.e. a sufficient condition). 
However, given an arbitrary pattern setup on an infinite description space, $\mathbb{P}^{\triangledown}$ is not always guaranteed to be a pattern structure. Theorem \ref{thm:antichaincompletion} gives a necessary and sufficient condition on $\mathbb{P}$ that makes $\mathbb{P}^{\triangledown}$  a \emph{pattern structure}. Here $ext^{\triangledown}$ and $int^{\triangledown}$ denote extent and intent of $\mathbb{P}^{\triangledown}$,  respectively.

\begin{theorem}\label{thm:antichaincompletion}
Let $\mathbb{P}=(\mathcal{G}, \underline{\mathcal{D}}, \delta)$ be a pattern setup, the \emph{antichain completion} of $\mathbb{P}$ is a \emph{pattern structure} \emph{if and only if} $\mathbb{P}$ is a \emph{pattern \multistructure{}}. Moreover:
\begin{eqnarray*}
(\forall S \in \mathscr{A}(\mathcal{D})) \: ext^{\triangledown}(S) = \bigcap ext[S]   & \phantom{:::} &
(\forall A \subseteq \mathcal{G}) \: int^{\triangledown}(A) = \minf(\delta[A]) = cov^*(A)
\end{eqnarray*} 

We have $\mathbb{P}^{\triangledown}_{ext} = \left\{\bigcap S \mid S \subseteq \mathbb{P}_{ext}\right\}$ and $\bigcap \emptyset = \mathcal{G} \in \mathbb{P}^{\triangledown}_{ext}$. 
\end{theorem}

\begin{mdframed}
\input{proofs/proof19.tex}
\end{mdframed}


\subsection{On Pattern Setups Direct Completions}

There is a completion that transforms any pattern setup to a pattern structure without demanding any additional property. 

\begin{theorem}\label{thm:directCompletion}
The \empha{direct completion} of $\mathbb{P}=(\mathcal{G}, \underline{\mathcal{D}}, \delta)$ is the \underline{\emph{pattern structure}}:
\begin{eqnarray*}
\mathbb{P}^{\blacktriangledown}  & = &  (\mathcal{G}, (\wp(\mathcal{D}), \subseteq), \delta^{\blacktriangledown} :g \mapsto \;\downarrow \delta(g)) 
\end{eqnarray*}
Where $(\forall S \subseteq \mathcal{D}) \: ext^{\blacktriangledown}(S) = \bigcap ext[S]$ and $(\forall A \subseteq \mathcal{G}) \: int^{\blacktriangledown}(A) = cov(A) = \delta[A]^\ell$. The set of definable sets is given by $\mathbb{P}^{\blacktriangledown}_{ext} = \left\{\bigcap S \mid S \subseteq \mathbb{P}_{ext}\right\}$.
\end{theorem}

\begin{mdframed}
\input{proofs/proof20.tex}

\end{mdframed}

\begin{example}\label{example:conceptlatticeofantichaincompletion}
\figref{fig:completionofrexample} depicts the concept lattice $\underline{\mathfrak{B}}(\mathbb{P}^\triangledown)$ of the \emph{antichain completion} of the pattern \multistructure{} $\mathbb{P}$ considered in \figref{fig:exampleSequences} (\emph{i.e.,} the description space is augmented with the top element $\top$). For any concept $(A, B)$, descriptions $d \in B$ in \textbf{bold} are those which whose $ext(d) = A$. Please notice that while description $``c"$ has for extent $\{g_1,g_2,g_4\}$, description $``c"$ does belong to the concept related to the extent $\{g_2,g_4\}$. Another important remark, are the underlined concepts. They represent concepts that are related to the non definable sets in $\mathbb{P}$ but still definable in $\mathbb{P}^\triangledown$, i.e.   $\{g_1,g_2\}$ and $\{g_1,g_2,g_3,g_4\}$ in $\mathbb{P}^\triangledown_{ext} \backslash \mathbb{P}_{ext}$. For instance, consider the intent of $\{g_1,g_2\}$ in the completion, each pattern $d$ has extent $ext(d) \supsetneq \{g_1,g_2\}$. Extent $\{g_1,g_2,g_3,g_4\}$ is \emph{non-coverable} in $\mathbb{P}$ and thus $int^\triangledown(\{g_1,g_2,g_3,g_4\}) = max(cov(\{g_1,g_2,g_3,g_4\})) = max(\emptyset) = \emptyset$. 
\end{example}

\begin{figure}[t]
\centering 
\scalebox{1}{\input{figures/hassediagramSequenceCompletion.tex}}
\caption{Concept lattice $\underline{\mathfrak{B}}(\mathbb{P}^\triangledown)$\label{fig:completionofrexample}.}
\end{figure}
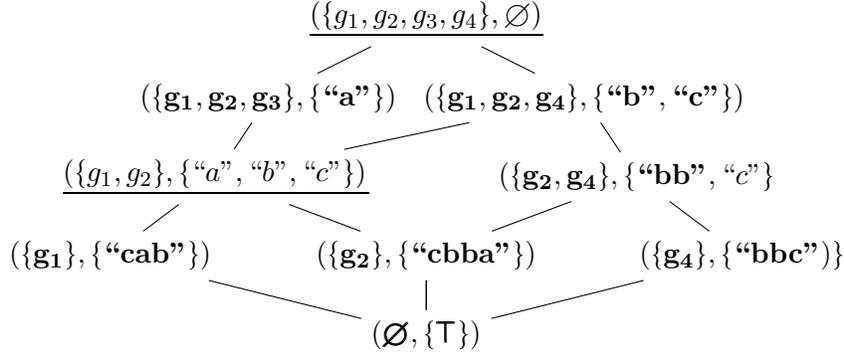

Note that while in Example \ref{example:conceptlatticeofantichaincompletion}, the size difference between the set of definable sets in the base pattern setup $\mathbb{P}_{ext}$ and the set of definable sets in the antichain completion $\mathbb{P}^\triangledown_{ext}$ is not large (i.e. $|\mathbb{P}^\triangledown_{ext} \backslash \mathbb{P}_{ext}|$ = 2). In some cases, the size of $\mathbb{P}^\triangledown_{ext}$ can be exponentially larger than $\mathbb{P}_{ext}$. Consider, for instance, the following example:

\begin{example}\label{example:exponentialcompletion}
Let $n \in \mathbb{N}$ with $n \geq 3$. We denote by $[n]$ the subset $[n] = \{1, 2, ..., n\}$. Let $\mathbb{P} = (\mathcal{G}, (\mathcal{D}, \subseteq), \delta)$ be the pattern setup  with $\mathcal{G} = \{g_i\}_{i \in [n]}$,  
\begin{eqnarray*}
\mathcal{D} = \{\{i\} \mid i \in [n]\} \cup \{[n] \backslash \{i\}  \mid i \in [n]\}
\end{eqnarray*}
and the mapping $\delta : g_i \mapsto [n] \backslash \{i\}$ for all $i \in [n]$. One can verify that we have $\mathbb{P}_{ext} = \{\{g_i\} \mid i \in [n]\} \cup \{\mathcal{G} \backslash \{g_i\} \mid i \in [n]\}$. Indeed, we have:
\begin{itemize}[topsep=1pt,noitemsep]
    \item  $(\forall i \in [n])$ $ext([n]\backslash\{i\}) = \{g_j \mid [n]\backslash\{j\} \subseteq \delta(g_j)\} = \{g_j \mid [n]\backslash\{j\} \subseteq [n]\backslash\{i\}\} = \{g_i\}$.
     \item  $(\forall i \in [n])$ $ext(\{i\}) = \{g_j \mid \{i\} \subseteq \delta(g_j)\} = \{g_j \mid \{i\} \subseteq [n]\backslash\{j\}\} = \mathcal{G} \backslash \{g_i\}$.
\end{itemize}
Hence, we have $|\mathbb{P}_{ext}| = 2 n$. However, according to Theorem \ref{thm:antichaincompletion}, we have
\begin{eqnarray*}
\mathbb{P}^\triangledown_{ext} = \left\{\bigcap S \mid S \subseteq \mathbb{P}_{ext}\right\} = \wp(\mathcal{G})
\end{eqnarray*}
since $\mathbb{P}_{ext}$ contains all the coatoms of $\wp(\mathcal{G})$ (i.e. $(\forall g \in \mathcal{G})$ $\mathcal{G} \backslash \{g\} \in \mathbb{P}_{ext}$) and the powerset lattice is coatomistic. It follows that $|\mathbb{P}^\triangledown_{ext}| = 2^n$. In other words, $\mathbb{P}^\triangledown_{ext}$ is exponentially larger than $\mathbb{P}_{ext}$. One should notice that the new description space associated to $\mathbb{P}^\triangledown$ is (order-)isomorphic to $(\wp([n]), \subseteq)$. 
\end{example}

%% file: proofs/proof22.tex
\begin{proof}
The case of $S = \emptyset$ is trivial since $\downarrow \emptyset = \emptyset$ and $ max(\emptyset) = \emptyset$.
Let be a nonempty set $S \subseteq P$ s.t. $(\exists C \in \mathscr{A}(P))\:S = \downarrow C$. Let us show that $C = max(S)$:
\begin{itemize}
    \item $C \subseteq max(S)$: let $c \in C \subseteq S$, suppose that $c \not\in max(S)$ that is $\exists x \in S$ s.t. $c < x$. Since $S = \downarrow C$ then $\exists c_2 \in C$ s.t. $x \leq c_2$. Thus $\exists c_2 \in C$ such that $c > c_2$ which is a contradiction with the fact that $C$ is an antichain.
    
    \item $C \supseteq max(S)$: Suppose $\exists a \in max(S)$ s.t. $a \not\in C$. We have $a \in S = \downarrow C$, that is: $\exists c \in C$ s.t. $a < c$ (since $a \not\in C$). However, since $S = \downarrow C$ then $C \subseteq S$. Thus, $\exists c \in S$ s.t. $a < c$ which is in contradiction with the fact that $a \in max(S)$.
\end{itemize}
This concluds the proof.
\end{proof}

%% file: figures/completelatticenotgoodantichaincompletion.tex
\begin{tikzpicture}[scale=.7]
\node (c0) at (0,1) {$c_0$};
\node (c1) at (0,2) {$c_1$};
\node (ci) at (0,3.5) {$\vdots$};
\node (cn) at (0,5) {$c_n$};
\node (cinf) at (0,6) {$\vdots$};

\node (top) at (0,8) {$\top$};

\node (a0) at (8,7) {$a_0$};
\node (a1) at (6.5,7) {$a_1$};
\node (ai) at (5,7) {$\hdots$};
\node (an) at (3.5,7) {$a_n$};
\node (ainf) at (1,7) {$\hdots$};

\node (b0) at (-8,7) {$b_0$};
\node (b1) at (-6.5,7) {$b_1$};
\node (bi) at (-5,7) {$\hdots$};
\node (bn) at (-3.5,7) {$b_n$};
\node (binf) at (-1,7) {$\hdots$};

\draw (c0) -- (c1) -- (ci) -- (cn) -- (cinf);

\draw (c0) -- (a0) -- (top);
\draw (c1) -- (a1) -- (top);
\draw (cn) -- (an) -- (top);

\draw (c0) -- (b0) -- (top);
\draw (c1) -- (b1) -- (top);
\draw (cn) -- (bn) -- (top);

\end{tikzpicture}

%% file: proofs/proof18.tex
\begin{proof}
Let us start by showing the first property that is if $(\mathscr{A}(P), \leqq)$ is a lattice then $(P, \leq)$ is \meetmultisemilattice{}. 
We have $(\mathscr{A}(P), \leqq)$ is a lattice. Then, following \cite{BoldiVigna2016}, we have:
\begin{eqnarray*}
\forall A \subseteq \mathscr{A}(P) & \exists C \in \mathscr{A}(P) & \downarrow A \cap \downarrow B = \downarrow C
\end{eqnarray*}
More generally:
\begin{eqnarray*}
\forall \mathbb{S} \subseteq \mathscr{A}(P) \text{ finite and nonempty } & \exists C \in \mathscr{A}(P) & \bigcap_{A \in \mathbb{S}}  \downarrow A = \downarrow C
\end{eqnarray*}
Since $C$ is an antichain, we have $C = max\left(\bigcap_{A \in \mathbb{S}}  \downarrow A\right)$ (Lemma \ref{lemma:antichainsandmaxelements}), that is:
\begin{eqnarray}\label{eq:propertyAntichainToMeet}
\forall \mathbb{S} \subseteq \mathscr{A}(P) \text{ finite and non empty }  & & \bigcap_{A \in \mathbb{S}}  \downarrow A = \downarrow max(\bigcap_{A \in \mathbb{S}}  \downarrow A)
\end{eqnarray}

Let be $S \subseteq P$ be a non empty finite subset. We need to show that $S$ has all its multi-infima. That is $S^\ell = \downarrow max(S^\ell)$. We have $S^\ell = \bigcap_{s \in S} \downarrow \{s\}$. Since,  $(\mathscr{A}(P), \leqq)$ is a lattice we have according to equation (\ref{eq:propertyAntichainToMeet}):
\begin{eqnarray*}
S^\ell = \bigcap_{s \in S}  \downarrow \{s\} = \downarrow max(S^\ell)
\end{eqnarray*}

For the second part of the proposition, consider that $\mathscr{A}(P)$ has a top elements. That is $\exists C \in \mathscr{A}(P)$ $P = \downarrow C$. Hence, we have $P = \downarrow max(P)$ since $C$ is an antichain (Lemma \ref{lemma:antichainsandmaxelements}). In other words, $\emptyset$ has all its multi-infima. 
\end{proof}

%% file: proofs/proof19.tex
\begin{proof}
%
Let us show that:
\begin{eqnarray*}
\mathbb{P} \text{ is a pattern \multistructure{} } & \Leftrightarrow & \mathbb{P}^{\triangledown} \text{ is a pattern structure}
\end{eqnarray*}
Recall that $\mathbb{P}^\triangledown$ is a pattern structure \textbf{iff} every subset of $\delta^\triangledown[\mathcal{G}]$ has a meet in $(\mathcal{A}(\mathcal{D}), \leq)$.  For $A \subseteq \mathcal{G}$ we have:
\begin{eqnarray*}
\delta^\triangledown[A]^\ell 
= \{S \in \mathcal{A}(\mathcal{D}) \mid (\forall g \in A)\,S \subseteq \downarrow \delta(g)\}
= \{S \in \mathcal{A}(\mathcal{D}) \mid \: S \subseteq \delta[A]^\ell\}
\end{eqnarray*}
where $\delta[A]^\ell$ and $\delta^\triangledown[A]^\ell$ denote respectively the lower bounds of $\delta[A]$ w.r.t. $\sqsubseteq$ and the lower bounds of of $\delta^\triangledown[A]$ w.r.t. $\leq$ (recall that $\delta[A]^\ell = \bigcap_{g \in A} \downarrow \delta(g)$). In this proof $\downarrow$ refers to the down-closure related to $\sqsubseteq$.

We show each implication independently:
\begin{itemize}[noitemsep,nolistsep]
\item ($\Rightarrow$) Let $A \subseteq \mathcal{G} : \delta[A]^\ell = \downarrow max(\delta[A]^\ell)$. Thus 
$
\delta^\triangledown[A]^\ell = \{S \in \mathcal{A}(\mathcal{D}) \mid \: S \subseteq \downarrow max(\delta[A]^\ell)\} 
= \{S \in \mathcal{A}(\mathcal{D}) \mid S\:\leq\:max(\delta[A]^\ell)\} 
$.
Since $max(\delta[A]^\ell) \in \mathcal{A}(\mathcal{D})$, so $max(\delta[A]^\ell)$ is the meet of $\delta^\triangledown[A]$ in $\mathcal{A}(\mathcal{D})$.

\item ($\Leftarrow$) $\mathbb{P}^{\triangledown}$ is a pattern structure is \textbf{equivalent} to say: $\forall A \subseteq \mathcal{G}$, $\delta^\triangledown[A]$ has a meet $M \in \mathcal{A}(\mathcal{D})$. That is, $\exists M \in \delta^\triangledown[A]^\ell$ for $A \subseteq \mathcal{G}$:
$\forall S \in \mathcal{A}(\mathcal{D}) :  S \subseteq \delta[A]^\ell \Leftrightarrow S \subseteq \downarrow M$.
Particularly, for $S = \{d\}$ with $d \in \mathcal{D}$, we deduce that
$
\forall d \in \delta[A]^\ell : d \in \downarrow M 
$. 
Thus, $\delta[A]^\ell \subseteq \downarrow M$. Moreover, since $M \subseteq \delta[A]^\ell$ ($M \in \delta^\triangledown[A]^\ell$) and $\downarrow$ is a closure operator on $(\wp(\mathcal{D}), \subseteq)$ we have by monotony $\downarrow M \subseteq \delta[A]^\ell \subseteq \downarrow M$ (note that $\downarrow \delta[A]^\ell = \delta[A]^\ell$). We conclude  that we have $\delta[A]^\ell = \downarrow M$. Using Lemma \ref{lemma:antichainsandmaxelements} we obtain
$\delta[A]^\ell = \downarrow max(\delta[A]^\ell)$. 
\end{itemize}
We conclude the equivalence.

\smallskip 
Let us now determine $int^{\triangledown}$ and $ext^{\triangledown}$. 
The previous proof has shown that for $A \subseteq \mathcal{G}$ the meet of $\delta^\triangledown[A]$ is $max(\delta[A])^\ell$. i.e.:
\begin{eqnarray*}
int^{\triangledown}(A) = max(\delta[A]^\ell) = cov^*(A)
\end{eqnarray*}

\smallskip 
For $ext^{\triangledown}$ operator, let $S \in \mathcal{A}(\mathcal{D})$. We have:
\begin{eqnarray*}
ext^{\triangledown}(S)
& = & \{g \in \mathcal{G} \mid S \leq \sigma(g)\}
 = \{g \in \mathcal{G} \mid S \subseteq \downarrow \delta(g)\} \\
& = & \{g \in \mathcal{G} \mid (\forall d \in S) \, d \sqsubseteq \delta(g)\} 
= \bigcap_{d \in S} ext(d) = \bigcap ext[S]
\end{eqnarray*}

Let us show that $\mathbb{P}^{\triangledown}_{ext} = \left\{\bigcap S \mid S \subseteq \mathbb{P}_{ext}\right\}$. By  definition of $ext^{\triangledown}$, the property $\mathbb{P}^\triangledown_{ext} \subseteq \left\{\bigcap S \mid S \subseteq \mathbb{P}_{ext}\right\}$ holds. For the inverse inclusion, it is sufficient to show that $\mathbb{P}_{ext} \subseteq \mathbb{P}^\triangledown_{ext}$ (since $(\mathbb{P}^\triangledown_{ext}, \subseteq)$ is closed under intersection). Let $A \in \mathbb{P}_{ext}$. $\exists d \in \mathcal{D}$ s.t. $A = ext(d)$. Since $\{d\} \in \mathcal{A}(\mathcal{D})$, and $ext^{\triangledown}(\{d\}) = ext(d) = A$. We conclude that $A \in \mathbb{P}^\triangledown_{ext}$. Hence, $\mathbb{P}^{\triangledown}_{ext} = \left\{\bigcap S \mid S \subseteq \mathbb{P}_{ext}\right\}$. 
\end{proof}

%% file: proofs/proof20.tex
\begin{proof}
Let us show that the pattern setup $\mathbb{P}^{\blacktriangledown}$ is a pattern structure. Let $A \subseteq \mathcal{G}$.  We need to show that $\delta^\triangledown[A]$ has a meet in $(\wp(\mathcal{D}), \subseteq)$. We have:
\begin{eqnarray*}
\delta^\triangledown[A]^\ell = \{S \subseteq \mathcal{D} \mid (\forall g \in A) \: S \subseteq \downarrow \delta(A)\} = \{S \subseteq \mathcal{D} \mid S \subseteq \delta[A]^\ell\}
\end{eqnarray*}
Since $\delta[A]^\ell \in \delta^\triangledown[A]^\ell$, we conclude $\delta[A]^\ell$ is the meet of $\delta^\triangledown[A]^\ell$, Hence: 
\begin{eqnarray*}
int^{\blacktriangledown}(A) = \delta[A]^\ell = cov(A)
\end{eqnarray*}

For the extent operator $ext^{\blacktriangledown}$, let $S \in \wp(\mathcal{D})$. We have
\begin{eqnarray*}
ext^{\blacktriangledown}(S)
= \{g \in \mathcal{G} \mid S \subseteq \downarrow \delta(g)\}
= \{g \in \mathcal{G} \mid (\forall d \in S) \, d \sqsubseteq \delta(g)\} 
= \bigcap ext[S]
\end{eqnarray*}

\smallskip
Let us show that $\mathbb{P}^{\blacktriangledown}_{ext} = \left\{\bigcap S \mid S \subseteq \mathbb{P}_{ext}\right\}$. Thanks to the definition of  $ext^{\blacktriangledown}$, property $\mathbb{P}^\triangledown_{ext} \subseteq \left\{\bigcap S \mid S \subseteq \mathbb{P}_{ext}\right\}$ holds. For the inverse inclusion, it is sufficient to show that $\mathbb{P}_{ext} \subseteq \mathbb{P}^\blacktriangledown_{ext}$ (since $(\mathbb{P}^\blacktriangledown_{ext}, \subseteq)$ is closed under intersection). Let $A \in \mathbb{P}_{ext}$, $\exists d \in \mathcal{D}$ s.t. $A = ext(d)$. 
We have
$ext^\blacktriangledown(\{d\})
= \{g \in \mathcal{G} \mid \{d\} \subseteq \downarrow \delta(g)\}
 = \{g \in \mathcal{G} \mid d \sqsubseteq \delta(g)\}
= ext(d) = A
$.
We conclude that $A \in \mathbb{P}^\blacktriangledown_{ext}$ and $\mathbb{P}^{\blacktriangledown}_{ext} = \left\{\bigcap S \mid S \subseteq \mathbb{P}_{ext}\right\}$. 
\end{proof}

%% file: figures/hassediagramSequenceCompletion.tex
\begin{tikzpicture}[scale=.7]
\node (0) at (-0,-1.5) {$(\pmb{\emptyset}, \{\pmb{\top}\})$};
\node (1) at (-6,0) {$(\mathbf{\{g_1\}}, \{\textbf{``cab"}\})$};
\node (2) at (0,0) {$(\mathbf{\{g_2\}}, \{\textbf{``cbba"}\})$};
\node (4) at (6,0) {$(\mathbf{\{g_4\}}, \{\textbf{``bbc"})\}$};
\node (12) at (-4,1.5) {$\underline{(\{g_1,g_2\}, \{``a", ``b", ``c"\})}$};
\node (24) at (4,1.5) {$(\mathbf{\{g_2,g_4\}}, \{\textbf{``bb"}, ``c"\}$};
\node (123) at (-3,3) {$(\mathbf{\{g_1,g_2,g_3\}}, \{\textbf{``a"}\})$};
\node (124) at (3,3) {$(\mathbf{\{g_1,g_2,g_4\}}, \{\textbf{``b"},\textbf{``c"}\})$};
\node (1234) at (0,4.5) {$\underline{(\{g_1,g_2,g_3,g_4\}, \emptyset)}$};

\draw (0) -- (1) -- (12) -- (123) -- (1234) -- (124) -- (24) -- (2) -- (0);
\draw (0) -- (2) -- (12) -- (124);
\draw (0) -- (4) -- (24);
\end{tikzpicture}

%% file: sections/conclusion.tex
\section{Conclusion and Discussion}
\label{sec:conclusiondiscussion}
In this paper, we provided a better understanding of the pattern setup framework. We have shown that while pattern structures demand a strong condition on the partially ordered set of descriptions, pattern setups do not require any additional property on the description space, which makes them rather permissive. We have introduced a new framework, namely pattern \multistructure{}, lying between both structures. Informally, pattern \multistructure s demands that the set of maximal common description resume properly the set of common descriptions of any subset of objects. Analogously to pattern structures, pattern \multistructure s are tightly linked to \multilattice s. We have shown also that the usual \emph{antichain completion} used in FCA literature to transform pattern setups, like sequence of itemsets ones, to pattern structure is applicable if and only if the considered pattern setup is a pattern \multistructure{}.  

An important open problem we are thoroughly working on is the following: 
``Given an arbitrary pattern setup $\mathbb{P}$ with a finite set of objects, generate its definable sets exhaustively and irredundantly". If the pattern setup is a pattern structure, the literature abounds of algorithms solving such a problem (e.g. \citep{ganter84,Kuznetsov1993,DBLP:journals/isci/OutrataV12}). However, no algorithm exists to solve such a problem for an arbitrary pattern setup. Indeed, the usual solution in FCA is to transform the pattern setup to a pattern structure via a completion (e.g. antichain completion for \citep{DBLP:conf/ilp/KuznetsovS05,DBLP:journals/ijgs/BuzmakovEJKNR16,DBLP:conf/icfca/CodocedoBKBN17}). Such a solution could create a (exponentially) larger search space. Other algorithms in the literature tackles this problem by enumerating support-closed patterns (e.g. \cite{DBLP:conf/sdm/YanHA03}). These latter algorihms may generate twice the same definable sets, since two support-closed patterns could have the same extent. Solving this problem will be the subject of future research.

%% file: IJGS2019.bbl
\begin{thebibliography}{36}
\newcommand{\enquote}[1]{``#1''}
\providecommand{\natexlab}[1]{#1}
\providecommand{\url}[1]{\normalfont{#1}}
\providecommand{\urlprefix}{}

\bibitem[Agrawal and Srikant(1995)]{agrawal1995mining}
Agrawal, Rakesh, and Ramakrishnan Srikant. 1995. ``Mining sequential
  patterns.'' In \emph{Data Engineering}, 3--14. IEEE.

\bibitem[Baixeries, Kaytoue, and Napoli(2012)]{DBLP:conf/cla/BaixeriesKN12}
Baixeries, Jaume, Mehdi Kaytoue, and Amedeo Napoli. 2012. ``Computing
  Functional Dependencies with Pattern Structures.'' In \emph{{CLA}}, Vol. 972,
  175--186.

\bibitem[Baixeries, Kaytoue, and
  Napoli(2014)]{DBLP:journals/amai/BaixeriesKN14}
Baixeries, Jaume, Mehdi Kaytoue, and Amedeo Napoli. 2014. ``Characterizing
  functional dependencies in formal concept analysis with pattern structures.''
  \emph{Ann. Math. Artif. Intell.} 72 (1-2): 129--149.

\bibitem[Belfodil, Kuznetsov, and Kaytoue(2018)]{DBLP:conf/cla/BelfodilKK18}
Belfodil, Aimene, Sergei~O. Kuznetsov, and Mehdi Kaytoue. 2018. ``Pattern
  Setups and Their Completions.'' In \emph{{CLA}}, Vol. 2123, 243--253.

\bibitem[Belfodil et~al.(2017)]{DBLP:conf/ijcai/BelfodilKRK17}
Belfodil, Aimene, Sergei~O. Kuznetsov, C{\'{e}}line Robardet, and Mehdi
  Kaytoue. 2017. ``Mining Convex Polygon Patterns with Formal Concept
  Analysis.'' In \emph{{IJCAI}}, 1425--1432.

\bibitem[Benado(1955)]{benado1955ensembles}
Benado, Mihail. 1955. ``Les ensembles partiellement ordonn{\'e}s et le
  th{\'e}or{\`e}me de raffinement de Schreier I Th{\'e}orie des
  multistructures.'' \emph{Czechoslovak Math. J.} 5 (3): 308--344.

\bibitem[Boldi and Vigna(2016)]{BoldiVigna2016}
Boldi, Paolo, and Sebastiano Vigna. 2016. ``On the lattice of antichains of
  finite intervals.'' \emph{CoRR} abs/1510.03675.
  \urlprefix\url{https://arxiv.org/abs/1510.03675}.

\bibitem[Boley et~al.(2010)]{DBLP:journals/tcs/BoleyHPW10}
Boley, Mario, Tam{\'{a}}s Horv{\'{a}}th, Axel Poign{\'{e}}, and Stefan Wrobel.
  2010. ``Listing closed sets of strongly accessible set systems with
  applications to data mining.'' \emph{Theor. Comput. Sci.} 411 (3): 691--700.

\bibitem[Buzmakov et~al.(2016)]{DBLP:journals/ijgs/BuzmakovEJKNR16}
Buzmakov, Aleksey, Elias Egho, Nicolas Jay, Sergei~O. Kuznetsov, Amedeo Napoli,
  and Chedy Ra{\"{\i}}ssi. 2016. ``On mining complex sequential data by means
  of {FCA} and pattern structures.'' \emph{Int. J. General Systems} 45 (2):
  135--159.

\bibitem[Buzmakov, Kuznetsov, and Napoli(2015)]{DBLP:conf/icfca/BuzmakovKN15}
Buzmakov, Aleksey, Sergei~O. Kuznetsov, and Amedeo Napoli. 2015. ``Revisiting
  Pattern Structure Projections.'' In \emph{{ICFCA}}, Vol. 9113 of
  \emph{Lecture Notes in Computer Science}, 200--215.

\bibitem[Codocedo et~al.(2017)]{DBLP:conf/icfca/CodocedoBKBN17}
Codocedo, V{\'{\i}}ctor, Guillaume Bosc, Mehdi Kaytoue, Jean{-}Fran{\c{c}}ois
  Boulicaut, and Amedeo Napoli. 2017. ``A Proposition for Sequence Mining Using
  Pattern Structures.'' In \emph{{ICFCA}}, .

\bibitem[Codocedo and Napoli(2014)]{DBLP:conf/ecai/CodocedoN14}
Codocedo, V{\'{\i}}ctor, and Amedeo Napoli. 2014. ``Lattice-based biclustering
  using Partition Pattern Structures.'' In \emph{{ECAI} 2014}, Vol. 263,
  213--218.

\bibitem[Cordero et~al.(2004)]{cordero2004new}
Cordero, Pablo, Gloria Guti{\'e}rrez, Javier Mart{\'\i}nez, and
  Inmaculada~Perez de~Guzm{\'a}n. 2004. ``A new algebraic tool for automatic
  theorem provers.'' \emph{Annals of Mathematics and Artificial Intelligence}
  42 (4): 369--398.

\bibitem[Crampton and Loizou(2001)]{crampton2001completion}
Crampton, Jason, and George Loizou. 2001. ``The completion of a poset in a
  lattice of antichains.'' \emph{International Mathematical Journal} 1 (3):
  223--238.

\bibitem[Davey and Priestley(2002)]{davey2002introduction}
Davey, Brian~A, and Hilary~A Priestley. 2002. \emph{Introduction to lattices
  and order}. Cambridge university press.

\bibitem[Ferr{\'{e}} and Ridoux(2000)]{DBLP:conf/iccs/FerreR00}
Ferr{\'{e}}, S{\'{e}}bastien, and Olivier Ridoux. 2000. ``A Logical
  Generalization of Formal Concept Analysis.'' In \emph{{ICCS}}, Vol. 1867 of
  \emph{Lecture Notes in Computer Science}, 371--384.

\bibitem[Ganter(1984)]{ganter84}
Ganter, Bernhard. 1984. ``Two basic algorithms in concept analysis.''
  \emph{Technical report, Technische Hoschule Darmstadt} .

\bibitem[Ganter and Kuznetsov(2001)]{DBLP:conf/iccs/GanterK01}
Ganter, Bernhard, and Sergei~O. Kuznetsov. 2001. ``Pattern Structures and Their
  Projections.'' In \emph{{ICCS}}, Vol. 2120 of \emph{LNCS}, 129--142.
  Springer.

\bibitem[Ganter and Wille(1989)]{ganter1989conceptual}
Ganter, Bernhard, and Rudolf Wille. 1989. ``Conceptual scaling.'' In
  \emph{Applications of combinatorics and graph theory to the biological and
  social sciences}, 139--167.

\bibitem[Ganter and Wille(1999)]{GanterW99}
Ganter, Bernhard, and Rudolf Wille. 1999. \emph{{Formal Concept Analysis}}.
  Springer.

\bibitem[Garriga, Kralj, and Lavrac(2008)]{DBLP:journals/jmlr/GarrigaKL08}
Garriga, Gemma~C., Petra Kralj, and Nada Lavrac. 2008. ``Closed Sets for
  Labeled Data.'' \emph{Journal of Machine Learning Research} 9: 559--580.

\bibitem[Kaytoue, Kuznetsov, and Napoli(2011)]{DBLP:conf/ijcai/KaytoueKN11}
Kaytoue, Mehdi, Sergei~O. Kuznetsov, and Amedeo Napoli. 2011. ``{Revisiting
  Numerical Pattern Mining with Formal Concept Analysis}.'' In \emph{IJCAI},
  1342--1347.

\bibitem[Kuznetsov(1993)]{Kuznetsov1993}
Kuznetsov, Sergei~O. 1993. ``{A Fast Algorithm for Computing All Intersections
  of Objects in a Finite Semi-lattice}.'' \emph{{Nauchno-Tekhnicheskaya
  Informatsiya}} ser. 2 (1): 17--20.

\bibitem[Kuznetsov(1999)]{DBLP:conf/pkdd/Kuznetsov99}
Kuznetsov, Sergei~O. 1999. ``Learning of Simple Conceptual Graphs from Positive
  and Negative Examples.'' In \emph{{PKDD}}, 384--391.

\bibitem[Kuznetsov(2009)]{DBLP:conf/rsfdgrc/Kuznetsov09}
Kuznetsov, Sergei~O. 2009. ``Pattern Structures for Analyzing Complex Data.''
  In \emph{RSFDGrC}, 33--44.

\bibitem[Kuznetsov and Samokhin(2005)]{DBLP:conf/ilp/KuznetsovS05}
Kuznetsov, Sergei~O., and Mikhail~V. Samokhin. 2005. ``Learning Closed Sets of
  Labeled Graphs for Chemical Applications.'' In \emph{{ILP}}, Vol. 3625 of
  \emph{LNCS}, 190--208.

\bibitem[Lumpe and Schmidt(2015)]{DBLP:conf/cla/LumpeS15}
Lumpe, Lars, and Stefan~E. Schmidt. 2015. ``Pattern Structures and Their
  Morphisms.'' In \emph{{CLA}}, 171--179.

\bibitem[Mart{\'\i}nez et~al.(2001)]{martinez2001multilattices}
Mart{\'\i}nez, J, G~Guti{\'e}rrez, IP~de~Guzm{\"a}n, and P~Cordero. 2001.
  ``Multilattices via multisemilattices.'' \emph{Topics in applied and
  theoretical mathematics and computer science} 238--248.

\bibitem[Mart{\'{\i}}nez et~al.(2005)]{DBLP:journals/dm/MartinezGGC05}
Mart{\'{\i}}nez, Javier, Gloria Guti{\'{e}}rrez, Inmaculada~Perez
  de~Guzm{\'{a}}n, and Pablo Cordero. 2005. ``Generalizations of lattices via
  non-deterministic operators.'' \emph{Discrete Mathematics} 295 (1-3):
  107--141.

\bibitem[Outrata and Vychodil(2012)]{DBLP:journals/isci/OutrataV12}
Outrata, Jan, and Vil{\'{e}}m Vychodil. 2012. ``Fast algorithm for computing
  fixpoints of Galois connections induced by object-attribute relational
  data.'' \emph{Inf. Sci.} 185 (1): 114--127.

\bibitem[Pawlak(1982)]{DBLP:journals/ijpp/Pawlak82}
Pawlak, Zdzislaw. 1982. ``Rough sets.'' \emph{International Journal of Parallel
  Programming} 11 (5): 341--356.

\bibitem[Roman(2008)]{roman2008lattices}
Roman, S. 2008. \emph{Lattices and Ordered Sets}. Springer New York.

\bibitem[Wang and Han(2004)]{DBLP:conf/icde/WangH04}
Wang, Jianyong, and Jiawei Han. 2004. ``{BIDE:} Efficient Mining of Frequent
  Closed Sequences.'' In \emph{{ICDE}}, 79--90. {IEEE} Computer Society.

\bibitem[Wille(1982)]{wille1982restructuring}
Wille, Rudolf. 1982. ``Restructuring lattice theory: an approach based on
  hierarchies of concepts.'' In \emph{Ordered sets}, 445--470.

\bibitem[Yan, Han, and Afshar(2003)]{DBLP:conf/sdm/YanHA03}
Yan, Xifeng, Jiawei Han, and Ramin Afshar. 2003. ``CloSpan: Mining Closed
  Sequential Patterns in Large Datasets.'' In \emph{{SDM}}, 166--177. {SIAM}.

\bibitem[Zaki(2001)]{DBLP:journals/ml/Zaki01}
Zaki, Mohammed~Javeed. 2001. ``{SPADE:} An Efficient Algorithm for Mining
  Frequent Sequences.'' \emph{Machine Learning} 42 (1/2): 31--60.

\end{thebibliography}
